\documentclass[submission,copyright,creativecommons]{eptcs}

\providecommand{\event}{GandALF 2023} % Name of the event you are submitting to

\usepackage{iftex}

\ifpdf
  \usepackage{underscore}         % Only needed if you use pdflatex.
  \usepackage[T1]{fontenc}        % Recommended with pdflatex
\else
  \usepackage{breakurl}           % Not needed if you use pdflatex only.
\fi

\usepackage[utf8]{inputenc}
\usepackage{amsmath}
\usepackage{amssymb}
\usepackage{mathtools}
\usepackage[LGR,T1]{fontenc}
\usepackage[utf8]{inputenc}

\usepackage[english]{babel}
\usepackage{booktabs}
\usepackage{pgfplots}
\usepackage{multirow}
\usepackage{comment}
\usepackage{pgfplots}
\usepackage{array}
\usetikzlibrary{pgfplots.groupplots}
\pgfplotsset{compat=newest}
\usepackage{float}
\usepackage[ruled,noline,noend,linesnumbered]{algorithm2e}
\usepackage[nounderscore]{syntax}

 \newcommand{\LPH}{{\textit{TPhL}}} 

\SetCommentSty{mycommfont}
\SetKw{KwAnd}{and}
\SetKw{KwNot}{not}
\SetKw{KwBreak}{break}
\SetKw{KwInfl}{in}
\SetKwProg{Pn}{}{:}{\KwRet}

\usepackage{afterpage}

\usepackage{pgf}
\usepackage{tikz}
\usetikzlibrary{arrows,shapes,automata,backgrounds}

\let\emptyset\varnothing

\definecolor{mygray}{gray}{0.6}

\def\plstate{\textsf{\small state\_phenomenon}}

\def\plstart{\textsf{\small start}}
\def\plend{\textsf{\small end}}

\def\plbefore{\textsf{\small before}}
\def\plmeets{\textsf{\small meets}}
\def\ploverlaps{\textsf{\small overlaps}}
\def\plfinishes{\textsf{\small finishes}}
\def\plcontains{\textsf{\small contains}}

\def\plstarts{\textsf{\small starts}}
\def\plequals{\textsf{\small equals}}

\def\plfilter{\textsf{\small filter}}

\newcommand{\mt}[1]{{\mathit{#1}}}

\DeclareMathSymbol{\mlq}{\mathord}{operators}{``}
\DeclareMathSymbol{\mrq}{\mathord}{operators}{`'}

\newenvironment{mysplit}%
  {\arraycolsep 0pt \begin{array}{l}}%
  {\end{array}}

\usepackage{mathtools}

\usepackage{xcolor}

\makeatletter
\newcommand*\bigcdot{\mathpalette\bigcdot@{.5}}
\newcommand*\bigcdot@[2]{\mathbin{\vcenter{\hbox{\scalebox{#2}{$\m@th#1\bullet$}}}}}
\makeatother
\usepackage{amsthm}
\usepackage{scalefnt}

\usepackage{stackengine}

\newtheorem{theorem}{Theorem}
\DeclareMathAlphabet{\mathcal}{OMS}{cmsy}{m}{n}

\title{Handling of Past and Future with Phenesthe+}
\author{Manolis Pitsikalis \and Alexei Lisitsa \and Patrick Totzke\\
\institute{Department of Computer Science, University of Liverpool, United Kingdom}\\
\email{\{e.pitsikalis, a.lisitsa, totzke\}@liverpool.ac.uk}}

\newcommand{\authorrunning}{M.~Pitsikalis, A.~Lisitsa and P.~Totzke}
\newcommand{\titlerunning}{Handling of Past and Future with Phenesthe+}

\hypersetup{
  bookmarksnumbered,
  pdftitle    = {\titlerunning},
  pdfauthor   = {\authorrunning},
  pdfsubject  = {Temporal logics and complex event processing},               % Consider adding a more appropriate subject or description
}

\begin{document}

\maketitle
\begin{abstract}
    Writing temporal logic formulae for properties that combine instantaneous events with overlapping temporal phenomena of some duration is difficult in classical temporal logics. To address this issue, in previous work we introduced a new temporal logic with intuitive temporal modalities specifically tailored for the representation of both instantaneous and durative phenomena. We also provided an implementation of a complex event processing system, Phenesthe, based on this logic,  that has been applied and tested on a real maritime surveillance scenario. 
    
    In this work, we extend our temporal logic with two extra modalities to increase its expressive power for handling future formulae. We compare the expressive power of different fragments of our logic with Linear Temporal Logic and dyadic first-order logic. Furthermore, we define correctness criteria for stream processors that use our language. Last but not least, we evaluate  empirically the performance of Phenesthe+, our extended implementation, and show that the increased expressive power does not affect efficiency significantly.
\end{abstract}
\section{Introduction}
Temporal logics are widely used in many domains as they allow the formalisation of time dependent properties. For example in philosophy temporal logics can be used to reason about issues involving the temporal domain~\cite{OHRSTROM2006447}, in computer science and specifically, in monitoring, temporal logics are used to specify and monitor specific properties of a system~\cite{Mascle_2020,Bauer_2011}, in complex event processing or recognition~\cite{Cugola_Margara_2010,Anicic_2010,Beck_2018}---which is the focus of this work---temporal logics are used for specifying temporal phenomena and detecting them in streams of information. Naturally, each temporal logic comes with its own focus and limitations. 

The starting point for many monitoring systems is Linear Temporal Logic (LTL)~\cite{Pnueli_1977}, where formulae are interpreted over single event sequences. This makes it difficult to incorporate concurrent activities such as, for instance, those carried out by a chef following some recipe to create a meal (see Figure~\ref{fig:foodex}). They may prepare multiples dishes of the same course in parallel, however the preparation of each dish happens on different overlapping intervals. This is difficult to formalise with logics that talk about single traces of events. 
Temporal logics that allow the representation of concurrent activities directly are those of Halpern and Shoham (HS)~\cite{halpern1991} and Allen's Algebra~\cite{Allen_1983}. Both logics are interpreted over sets of (possibly overlapping, discrete) time intervals, and model instantaneous events via point intervals ([t,t]).

\begin{figure}
\centering
\includegraphics[width=0.6\textwidth]{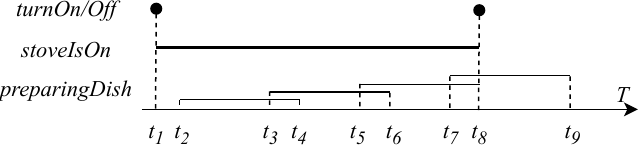}
\caption{Example of instantaneous and durative temporal phenomena. $\mt{turnOn/Off}$ is instantaneous and true when the stove is turned on/off, $\mt{stoveIsOn}$ is durative and holds when the stove is on, and finally $\mt{preparingDish}$ is also durative and holds when a dish is being prepared.}
\label{fig:foodex}
\end{figure}

In previous work~\cite{pitsikalis22a}, we introduced a temporal logic
similar to those of  Halpern/Shoham and Allen, which is interpreted over separate (sets of) time intervals and instantaneous events, which was specifically designed for a maritime application domain. It allows to easily specify concurrent activities and related start and endpoints.
Deliberately absent in our logics (and related prior ones) are explicit negations/complements of formulae that hold on overlapping intervals. 
We implemented an event processing system\footnote{\url{https://manospits.github.io/Phenesthe/}} and evaluated it on a real maritime monitoring scenario~\cite{pitsikalis22b}---whereby maritime experts authored maritime phenomena of interest---proving that our system is capable of producing phenomena detections in real time. 

In this work, we extend our language with two extra temporal modalities: minimal range ($\looparrowright$), and interval filtering ($\plfilter)$. The minimal range operator allows to capture durative temporal phenomena that start at the latest occurrence of a starting condition before the stopping condition, for example, the last working period of machine since someone operated it until it broke down. Filtering is very important for specifying
the duration constraints
of a temporal phenomenon, e.g., a steak is cooked rare if its on a hot pan for approximately 90 seconds on each side. Both examples are impossible to formalise in the original version of our temporal logic. It is evident that $\looparrowright$ has similar semantics to the `until' operator of LTL. However, formulae that utilise `until' are true on instants of time, while formulae that utilise `$\looparrowright$' are true on intervals. While in terms of expressive power the two operators are similar, $\looparrowright$, in practice, allows efficient computations and more concise formulae. Similarly, in the case of filtering, writing an LTL formula for filtering periods based on some fixed threshold is possible, however the formula is not trivial and its length depends the threshold. Comparing LTL and our temporal logic, we show that the fragment of our original language that is comparable with LTL was at most as expressive as the pure past fragment of LTL, while the same fragment but with the addition of the `$\looparrowright$' and `$\plfilter$' operators has equal expressive power to LTL. Concerning our full temporal logic, a comparison with LTL is impossible since the structures on which their respective semantics are based, are incomparable, however we show that our language is expressible in Dyadic First Order logic (DFO) and is strictly less expressive. As expected, including temporal modalities that involve the future requires additional steps for complex event processing. In order to guarantee that our extended implementation, Phenesthe+, is correct, inspired from runtime monitoring and verification~\cite{Aceto2019,Bartocci2018},  we define criteria for \textit{proper} stream processors of our language and discuss  how Phenesthe+ conforms to them. Finally, we illustrate through experimental evaluation that the efficiency of Phenesthe+, is not significantly compromised. Therefore the contributions of this paper are:
\begin{itemize}
    \item We extend the expressiveness of our language for representing ``look ahead happenings'',
    \item We  formally study the expressive power of the temporal logic introduced in~\cite{pitsikalis22a} and its extension,
    \item We define criteria for \textit{proper} stream processors utilising our temporal logic,
    \item We showcase that our stream processing engine is capable of performing real-time complex event processing by adopting a maritime surveillance use-case.
\end{itemize}
The paper is organised as follows. First in Section~\ref{sec:language} we describe our temporal logic. Next, in Section~\ref{sec:maritime} we illustrate through examples inspired by the maritime domain the usage of the new temporal modalities. In Section~\ref{sec:expressiveness} we study the expressive power of our temporal logic. Then, in Section \ref{sec:stream} we describe the requirements a stream processor should satisfy for processing formulae of our language, while in Section~\ref{sec:evaluation} we empirically  evaluate Phenesthe+. Finally, in Section~\ref{sec:relateddisc} we present work related to ours, summarise, and discuss further directions.
\section{The Language of Phenesthe}
\label{sec:language}
The key components of our language are instantaneous events, durative disjoint states and durative, possibly non disjoint, dynamic temporal phenomena. In what follows, `temporal phenomena' includes all of the three aforementioned categories.

\textbf{Syntax.} Formally, our Temporal Phenomena Definition Language (\LPH) is described by the triplet $\langle \mathcal{P}^{esd},L,\Phi \rangle$, where $\mathcal{P}^{esd}$ is a predicate set defined by the union of the event, state or dynamic temporal phenomenon predicates sets (in symbols $\mathcal{P}^{e/s/d}$ resp.); $L$ is a set defined by the union of the set of the logical connectives $\{\wedge,\vee,\neg,\in\}$, the set of temporal operators, $\{\rightarrowtail, \looparrowright, \sqcup,\sqcap,\setminus, \plfilter_\square\}$ where the `$\square$' symbol may be one of the following symbols ${\{<,\geq,=\}}$, the set of temporal relations $\{\plbefore$, $\plmeets$, $\ploverlaps$, $\plfinishes$, $\plstarts$, $\plequals$, $\plcontains\}$ and finally the set of the $\{\plstart,\plend\}$ operators; $\Phi$ is the set of formulae defined by the union of the formulae sets $\Phi^{\bigcdot}$, $\Phi^-$ and $\Phi^=$. We assume that the set of predicate symbols includes those with atemporal and fixed semantics, such as arithmetic comparison operators etc., however for simplification reasons in what follows we omit their presentation. Formulae of $\Phi^{\bigcdot}$ describe instantaneous temporal phenomena, and formulae of $\Phi^-$ describe durative temporal phenomena that hold (are true) in disjoint maximal intervals, finally formulae of $\Phi^=$ describe durative temporal phenomena that may hold in non-disjoint intervals. Figure~\ref{fig:foodex} shows an example of an event ($\mt{turnOn/Off}$), a state ($\mt{stoveIsOn}$) and a dynamic temporal phenomenon ($\mt{preparingDish}$).

Therefore given a set of event, state and dynamic phenomena predicates $\mathcal{P}^{esd}$ the formulae of $\LPH$ are defined as follows:
\begin{align*}
    \phi &:= \phi^{\bigcdot}\ |\ \phi^-\ |\ \phi^=\\
    \phi^{\bigcdot} &:= P^e(a_1,...,a_k)\ |\ \neg \phi^{\bigcdot}\ |\ \phi^{\bigcdot}\ [\wedge,\vee]\ \phi^{\bigcdot}\ |\  \plstart(\phi^-)\ |\ \plend(\phi^-)\ |\ \phi^{\bigcdot} \in \phi^- \\
    \phi^- &:=  P^s(a_1,...,a_k)\ | \phi^{\bigcdot}\ [\rightarrowtail, \looparrowright]\  \phi^{\bigcdot}\ |\ \phi^-\ [\sqcup,\sqcap,\setminus]\ \phi^- | \phi^- \plfilter\ n\ (\text{where}\ n \in \mathbb{N}\cup\{\infty\})\\
    \phi^= &:= P^d(a_1,...,a_k)\ |\ [\phi^-,\phi^=]\ [\plmeets,\ploverlaps,\plequals]\ [\phi^-,\phi^=]\\
    &\quad\qquad|\ [\phi^{\bigcdot},\phi^-,\phi^=]\ [\plstarts,\plfinishes]\ [\phi^-,\phi^=]\\
    &\quad\qquad|\ [\phi^{-},\phi^=]\ \plcontains\ [\phi^{\bigcdot},\phi^-,\phi^=]\\
    &\quad\qquad|\ [\phi^{\bigcdot},\phi^-,\phi^=]\ \plbefore\ [\phi^{\bigcdot},\phi^-,\phi^=]
\end{align*}
where $a_1,...,a_k$ correspond to terms denoting atemporal properties.

\textbf{Semantics.} We assume time is discrete and represented by the natural numbers $T=\mathbb{N}$ ordered via the `$<$' relation. In what follows  we assume that for all models discussed in this paper time is represented by $\mathbb{N}$.
For the formulae sets $\Phi^{\bigcdot}, \Phi^-$ and $\Phi^=$ we define the model $\mathcal{M}=\langle T, I, <,V^{\bigcdot}, V^-, V^= \rangle$ where $V^{\bigcdot}:\mathcal{P}^e\rightarrow 2^T$,  $V^-:\mathcal{P}^s\rightarrow 2^I$, $V^=:\mathcal{P}^d\rightarrow 2^I$ are valuations, and $I=\{ [ts,te] :ts<te\ \text{and}\ ts,te \in T\}\cup\{[ts,\infty): ts\in T\}$ is the set of time intervals of $T$. 
Intervals of the form $[ts,\infty)$ denote that a phenomenon started being true at $ts$,  and continues being true forever.
Intervals of the form $[ts,te]$ denote that a phenomenon started being true at $ts$ and stopped being true at $te$. In what follows, we will use the abbreviated version for bounded quantifiers, i.e., $\forall_{>z_1}^{<z_2} x (...)$ denotes $\forall x\ (x>z_1 \wedge x<z_2) \rightarrow (...)$, and $\exists_{>z_1}^{<z_2} x (...)$ denotes $\exists x\ (x>z_1 \wedge x<z_2) \wedge (...)$. 

%%%%%%%%%%%%%%%%%%%%%%%%%%%%%%%%%%%%%%%%%%%%%%%%%%%
%Timepoints
%%%%%%%%%%%%%%%%%%%%%%%%%%%%%%%%%%%%%%%%%%%%%%%%%%%
Given a model $\mathcal{M}$, the validity of a formula $\phi\in\Phi^{\bigcdot}$ at a timepoint $t\in T$ (in symbols $\mathcal{M},t\models\phi$) is determined by the rules below, starting with the boolean connectives.
\begin{itemize}
\item $\mathcal{M},t\models P^e(a_1,...,a_n)$  iff $t\in$ $ V^{\bigcdot}(P^e(a_1,...,a_n))$.
\item $\mathcal{M},t\models \neg \phi$ iff  $\mathcal{M},t\not\models\phi$.
\item $\mathcal{M},t\models \phi [\wedge,\vee] \psi$  iff  $\mathcal{M},t\models \phi\ $[and, or]$ \ \mathcal{M},t\models \psi$.
\end{itemize}
Next we define the semantics for $\plstart, \plend$ and $\in$ which allow interaction between formulae of $\Phi^{\bigcdot}$ and $\Phi^-$ via the starting, ending, and intermediate points of intervals at which $\Phi^-$ formulae hold.
\begin{itemize}
\item $\mathcal{M},t\models \plstart(\phi)$  iff  $\exists te.\ \mathcal{M},[t,te] \models \phi$ or $\mathcal{M},[t,\infty)\models \phi$, where $\mathcal{M},[t,te] \models \phi$ denotes the validity of a formula $\phi\in\Phi^-$ at an interval $[t,te]$ as defined below.
\item $\mathcal{M},t\models \plend(\phi)$  iff  $\exists ts.\ \mathcal{M},[ts,t] \models \phi$.
\item $\mathcal{M},t\models \phi \in \psi$  iff $\mathcal{M},t\models \phi$ and $\exists^{\leq t} ts.\exists_{\geq t} te.\ \mathcal{M},[ts,te] \models \psi$.
\end{itemize}
%%%%%%%%%%%%%%%%%%%%%%%%%%%%%%%%%%%%%%%%%%%%%%%%%%%%
% Disjoint intervals
%%%%%%%%%%%%%%%%%%%%%%%%%%%%%%%%%%%%%%%%%%%%%%%%%%%%

 Given a model $\mathcal{M}$, the validity of a formula $\phi\in\Phi^-$ at a time interval $i=[ts,te]\in I$ (in symbols $\mathcal{M},[ts,te]\models\phi$) is defined as follows. We start with the $\looparrowright$ and $\rightarrowtail$ operators which allow specifying minimal or maximal intervals between instants where formulae of $\Phi^{\bigcdot}$ are true.
\begin{itemize}
\item $\mathcal{M},i\models P^s(a_1,...,a_n)$ iff $i\in $  $V^-(P^s(a_1,...,a_n))$.
\item $\mathcal{M},[ts,te]\models \phi \looparrowright \psi$ iff $\mathcal{M},ts\models \phi$ and $\mathcal{M},te\models \psi \wedge \neg \phi$ and $\forall_{>ts}^{<te} t.\bigl[\mathcal{M},t\not\models\phi$ and $\mathcal{M},t\not\models\psi \wedge \neg \phi\bigr]$.
Therefore, $\phi \looparrowright \psi$ holds for the intervals that start at the latest instant $ts$ at which $\phi$ is true and end at first instant $te$ after $ts$ where $\psi\wedge\neg\phi$ is true.
\item $\mathcal{M},[ts,te]\models \phi \rightarrowtail \psi$  iff 
$
         \mathcal{M},ts\models \phi$ and $\mathcal{M},te\models \psi \wedge \neg \phi$
      and  $\forall_{>ts}^{<te} t.\ \mathcal{M},t\not\models\psi \wedge \neg \phi$ and $\forall^{<ts}ts'.$ $\mathcal{M},ts'\models \phi \rightarrow \exists^{<ts}_{>ts'} te'.\ \mathcal{M},te'\models \psi \wedge \neg \phi 
$.

Essentially, $\phi \rightarrowtail \psi$ holds for the disjoint maximal intervals that start at the earliest instant $ts$ where $\phi$ is true and end at the earliest instant $te$  where $\psi$ is true and $\phi$ is false.
\item $\mathcal{M},[ts,\infty) \models \phi \rightarrowtail \psi$ iff $
         \mathcal{M},ts\models \phi$ and $ \forall_{>ts} t.\ \mathcal{M},t\not\models\psi \wedge \neg \phi$ and $\forall^{<ts}ts'.\ \mathcal{M},ts'\models \phi \rightarrow \exists^{<ts}_{>ts'} te'.$ $\mathcal{M},te\models \psi \wedge \neg \phi 
$.

 Therefore a formula  $\phi \rightarrowtail \psi$ may hold indefinitely if there does not exist an  instant after $ts$ at which $\psi\wedge\neg\phi$ is satisfied. For simplification reasons in the semantics below we omit intervals open at the right to infinity since they can be treated in a similar manner.
 \end{itemize}
We continue with the definition of semantics for $\sqcup, \sqcap$ and $\setminus$ which correspond to the usual set operations but for time intervals.
 \begin{itemize}
\item $\mathcal{M},[ts,te]\models \phi\ \sqcup\ \psi$ iff\footnote{A first order definition of the semantics of $\sqcup$ is possible but more lengthy.} 
\begin{itemize}
\item exists a sequence of length $k > 1$ of intervals $i_1,...,i_k \in I$ where $i_k=[ts_k,te_k]$, $ts=ts_1$ and $te=te_k$  such that:
\begin{enumerate}
\item $\forall\alpha\in [1,k-1]$: $te_\alpha\in i_{\alpha+1}$, $ts_\alpha < ts_{\alpha+1}$ and $te_\alpha < te_{\alpha+1}$  ,
\item $\forall\beta\in[1,k]$:  $\mathcal{M},[ts_\beta,te_\beta]\models\phi$ or $\mathcal{M},[ts_\beta,te_\beta]\models\psi$, and
\item $\nexists i_\gamma=[ts_\gamma,te_\gamma] \in I-\{i_1,...,i_k\}$ where  $\mathcal{M},[ts_\gamma,te_\gamma]\models\phi$ or $\mathcal{M},[ts_\gamma,te_\gamma]\models\psi$ and $ts_1 \in i_\gamma$ or $te_k \in i_\gamma$
\end{enumerate}
\item or, $\mathcal{M},[ts,te]\models\phi$ or $\mathcal{M},[ts,te]\models\psi$ and $\nexists i_\gamma=[ts_\gamma,te_\gamma] \in I-\{[ts,te]\}$ where $\mathcal{M},[ts_\gamma,te_\gamma]\models\phi$ or $\mathcal{M},[ts_\gamma,te_\gamma]\models\psi$  and $ts \in i_\gamma$ or $te \in  i_\gamma$.
\end{itemize}
For a sequence of intervals, conditions (1-2) ensure that intervals, at which $\phi$ or $\psi$ are valid,  overlap or touch will coalesce, while condition (3) ensures that the resulting interval is maximal. In the case of a single interval, the conditions ensure that at the interval $[ts,te]$  $\phi$ or $\psi$ is valid, and that $[ts,te]$ is maximal. In simple terms, the temporal union $\phi\ \sqcup\ \psi$ holds for the intervals where at least one of $\phi$ or $\psi$ hold. The above definition of temporal union follows the definitions of temporal coalescing presented in~\cite{bohlen1998,Dohr2018}.
\item $\mathcal{M},[ts,te]\models \phi\ \setminus\ \psi$  iff $\exists  [ts',te'] \in I$ where $\mathcal{M},[ts',te']\models \phi$, $[ts,te]\subseteq [ts',te']$ (i.e., $[ts,te]$ subinterval of $[ts',te']$), $\forall [ts_\psi,te_\psi]\in I$ where $\mathcal{M},[ts_\psi,te_\psi]\models \psi$,  $[ts,te]\cap [ts_\psi,te_\psi]=\emptyset$ and finally $[ts,te]$ is maximal. In plain language, the temporal difference of formulae $\phi,\psi$ holds for the maximal subintervals of the intervals at which $\phi$ holds but $\psi$ doesn't hold.  
\item $\mathcal{M},[ts,te] \models \phi\ \sqcap\ \psi$   iff $\exists  [ts_\phi,te_\phi], [ts_\psi, te_\psi] \in I$ where $\mathcal{M},[ts_\phi,te_\phi]\models \phi$, $\mathcal{M},[ts_\psi,te_\psi]\models \psi$ and $\exists [ts,te]\in I$ where $[ts,te] \subseteq [ts_\phi,te_\phi]$, $[ts,te] \subseteq [ts_\psi,te_\psi]$ and $[ts,te]$ is maximal. In other words, the temporal intersection of two formulae of $\Phi^-$ holds for the intervals at which both formulae hold.
\end{itemize}
We finish the semantics for formulae of $\Phi^-$, with the semantics of the $\plfilter$ operator, which allows specifying constraints on the length of intervals at which formulae of $\Phi^-$ hold.
\begin{itemize}
\item $\mathcal{M},[ts,te] \models \phi\ \plfilter_{\{<,\geq,=\}}\ n$  iff $\mathcal{M},[ts,te] \models \phi$ and $te-ts\ {\{<,\geq,=\}}\ n$.
\end{itemize}

%%%%%%%%%%%%%%%%%%%%%%%%%%%%%%%%%%%%%%%%%%%%%%%%%%%%
% Non disjoint intervals
%%%%%%%%%%%%%%%%%%%%%%%%%%%%%%%%%%%%%%%%%%%%%%%%%%%%
Due to space limitations and for this part only we will adopt point intervals to refer to instants. This will allow us to define the semantics for formulae of $\Phi^=$ without specifying different rules for involved sub-formulae of $\Phi^{\bigcdot}$. In other words given a $\phi\in \Phi^{\bigcdot}$ we will denote the satisfaction relation $\mathcal{M},t\models \phi$ as $\mathcal{M},[t,t]\models \phi$. Given a model $\mathcal{M}$, the validity of a formula $\phi\in\Phi^=$ at a time interval $[ts,te]\in I$ (in symbols $\mathcal{M},[ts,te]\models\phi$) is defined as follows:
\begin{itemize}
%%%%%%%%%%%%%%%%%%%%%%%%%%%%%%%
% activity predicate
%%%%%%%%%%%%%%%%%%%%%%%%%%%%%%%
\item $\mathcal{M},[ts,te]\models P^d(a_1,...,a_n)$ iff $[ts,te]\in V^=(P^d(a_1,...,a_n))$.
%%%%%%%%%%%%%%%%%%%%%%%%%%%%%%%
% before
%%%%%%%%%%%%%%%%%%%%%%%%%%%%%%%
\item  $\mathcal{M},[ts,te]\models \phi\ \plbefore\ \psi$   iff 
% \begin{itemize}
% %
% % case - -
% %
   $\exists te'.\exists_{>te'} ts'.\bigl[\mathcal{M},[ts,te']\models \phi$ and $\mathcal{M},[ts',te]\models \psi$ and $ \forall ts''. \forall_{>te'}^{<ts'} te''.$ $  \mathcal{M},[ts'',te''] \not\models \phi$ and $\forall_{>te'}^{<ts'} ts''. \forall te''. \mathcal{M},[ts'',te''] \not\models \psi\bigr]$.
In our approach the `before' relation holds only for intervals where the pair of instants or intervals at which the participating formulae are true or hold, are contiguous. For example, for the intervals $[1,2]$, $[1,3]$ and $[5,6]$ only $[1,3]$ is $\plbefore$ $[5,6]$. We chose to limit the intervals satisfying the $\plbefore$ relation, as in practice it is usually the case that the interval directly before another one is required for specifying a dynamic phenomenon. 
%%%%%%%%%%%%%%%%%%%%%%%%%%%%%%%
% meets
%%%%%%%%%%%%%%%%%%%%%%%%%%%%%%%
\item $\mathcal{M},[ts,te]\models \phi\ \plmeets\ \psi$   iff  $\exists t. \mathcal{M},[ts,t]\models \phi$ and $\mathcal{M},[t,te]\models \psi$.
%%%%%%%%%%%%%%%%%%%%%%%%%%%%%%%
% overlaps
%%%%%%%%%%%%%%%%%%%%%%%%%%%%%%%
 \item $\mathcal{M},[ts,te]\models \phi\ \ploverlaps\ \psi$   iff $\exists_{>ts}^{<te} ts'.\exists_{>ts'}^{<te} te'.\bigl[ \mathcal{M},[ts,te']\models \phi$ and $\mathcal{M},[ts',te]\models \psi\bigr]$.
%%%%%%%%%%%%%%%%%%%%%%%%%%%%%%%
% finishes
%%%%%%%%%%%%%%%%%%%%%%%%%%%%%%%
\item $\mathcal{M},[ts,te]\models \phi\ \plfinishes\ \psi$  iff 
 $\exists_{>ts}^{\leq te} ts'.\bigl[\mathcal{M},[ts',te]\models \phi$ and $\mathcal{M},[ts,te]\models \psi\bigr]$.
%%%%%%%%%%%%%%%%%%%%%%%%%%%%%%%
% starts
%%%%%%%%%%%%%%%%%%%%%%%%%%%%%%%
\item $\mathcal{M},[ts,te] \models\phi\ \plstarts\ \psi$  iff 
  $\exists^{<te}_{\geq ts} te'.\bigl[\mathcal{M},[ts,te']\models \phi$ and $\mathcal{M},[ts,te]\models \psi\bigr]$.
%%%%%%%%%%%%%%%%%%%%%%%%%%%%%%%
% equals
%%%%%%%%%%%%%%%%%%%%%%%%%%%%%%%
\item $\mathcal{M},[ts,te]\models \phi\ \plequals\ \psi$   iff  $\mathcal{M},[ts,te]\models \phi$ and $\mathcal{M},[ts,te]\models \psi$.
%%%%%%%%%%%%%%%%%%%%%%%%%%%%%%%
% contains
%%%%%%%%%%%%%%%%%%%%%%%%%%%%%%%
\item $\mathcal{M},[ts,te]\models \phi\ \plcontains\ \psi$ iff $\mathcal{M},[ts,te]\models \phi$ and $\exists_{>ts} ts'.\exists^{<te} te'.$ $ \mathcal{M},[ts',te']\models \psi$.
\end{itemize}

\section{Examples of maritime properties expressed in \LPH}
\label{sec:maritime}
We demonstrate the usability of \LPH\ and the new temporal modalities by adopting a maritime monitoring scenario. When it comes to maritime surveillance there are several resources available; for example the Automatic Identification System (AIS) allows the transmission of timestamped positional and ancillary data from vessels, maritime areas in the form of polygons can be used for producing vessel-area relations and so on. Similar to~\cite{pitsikalis22a}, we assume the input consists of AIS messages along with spatial events relating vessels to areas of interest e.g., port areas, fishing areas and so on. Therefore our task here involves detecting maritime phenomena of interest i.e., the instants, time periods at which they are true over a maritime input stream. Below we formalise some maritime temporal phenomena\footnote{The complete set of definitions is available in our online repository \url{https://github.com/manospits/Phenesthe/tree/future}.} that utilise the new temporal modalities ($\looparrowright$ and $\plfilter$). 

\textbf{Fishing warning.} Illegal fishing is a very important issue. Vessels engaged in illegal fishing typically declare fake ship-types. Consider the formalisation below for detecting suspicious stops in fishing areas.
\begin{align*} 
&\plstate\ \mt{fishing\_warning(V,F)}:\\
&\quad ((\mt{in\_fishing\_area}(V,F) \wedge \neg \mt{vessel\_type(V, fishing)})\  \sqcap\ stopped(V))\ \plfilter_{\geq} 600.
\end{align*}
$\plstate$ is a keyword for declaring the phenomenon type, $\mt{in\_fishing\_area}$ is a user defined state that holds for the time periods a vessel $V$ is within a fishing area $F$, while $\mt{stopped}$ is a state that holds for the time periods a vessel is stopped. Finally, $\mt{vessel\_type}(V,T)$ is an atemporal predicate that is true when vessel $V$ has type $T$. Therefore, a vessel performs a $\mt{fishing\_warning}$, if it is not a fishing vessel, and it is stopped within a fishing area for a period longer that 10 min (600 sec). Here, filtering is used for minimising false detections occurring from AIS errors (e.g., zero speed) or normal activities.

\textbf{Port waiting time}. Monitoring the waiting time of vessels since they entered a port and until they get moored is highly useful for various operational and logistical reasons (e.g., efficient planning of resources). However, some vessels may enter and leave a port without mooring---due to weather conditions for example. We formalise port waiting time below.
\begin{align*} 
&\plstate\ \mt{waiting\_time}(V,P):\\
&\quad \plstart(\mt{in\_port}(V,P)) \looparrowright \plstart(\mt{moored}(V,P)).  
\end{align*}
$\mt{in\_port}$ is a state that holds when a vessel is in a port, while $\mt{moored}$ is a state that holds when a vessel is moored at a port. Note that the left and right arguments of $\looparrowright$ are formulae of $\Phi^{\bigcdot}$, therefore if these formulae were used in other definitions we could have defined corresponding events. Consequently, the $\mt{waiting\_time}$ state holds for the minimal periods between the time a vessel enters a port and the time the vessel starts being moored. Here we are interested in the minimal period, as we want to detect only the cases where a vessel entered a port and got moored.

\section{Expressiveness}
\label{sec:expressiveness}
In this section we study the expressive power of our language. We consider three syntactic fragments of \LPH. The first one, denoted as \LPH$_o^-$, corresponds to the original version of the language (w/o $\looparrowright,\plfilter$) and excluding formulae of $\Phi^=$ (recall that $\Phi^=$ formulae  hold on possibly non-disjoint intervals). The second is \LPH$^-$, which is the same as \LPH$_o^-$  but includes $\looparrowright$ and $\plfilter$, while the third, \LPH\,  corresponds to the complete language. Figure~\ref{fig:exprel} (left) illustrates the syntactic relation between \LPH, \LPH$^-$ and \LPH$^-_o$.  In more detail, we will show that \LPH$_o^-$ is equally expressive as pure past \textit{LTL}, \LPH$^-$ has equal expressive power to \textit{LTL}, and finally \LPH\ is strictly less expressive than \textit{DFO}. The relations in terms of expressive power between the  different language fragments are illustrated in Figure~\ref{fig:exprel} (right). 

\subsection{Preliminaries}
\label{sec:prel}
Before we continue with our analysis, as a reminder we present the syntax of LTL with past, and First Order Monadic Logic of Order (\textit{FOMLO}).
\begin{figure}[t]
    \centering
    \includegraphics[width=0.35\textwidth]{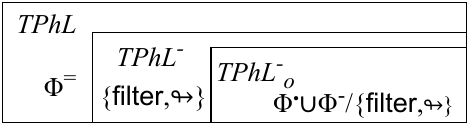}
    \hspace{30pt}
    \begin{tikzpicture}
    \node (lm) at (-2,-1) {\LPH$^-$};
    \node (lom) at (-2,-2) {\LPH$_o^-$};
    \node (dfo) at (2,-1) {\textit{DFO}};
    \node (l) at (2,-2) {\LPH};
    \node (tlp) at (0.0,-2) {\textit{LTL}$[\texttt{Y}\texttt{S}]$};
     \node (tl) at (0.0,-1) {\textit{LTL}$[\texttt{X}\texttt{U},\texttt{Y}\texttt{S}]$};
\draw [black,style=double] (lom) -- (tlp);
\draw [->,black] (tl) -- (tlp);
\draw [->,black] (lm) -- (lom);
\draw [black,style=double] (lm) -- (tl);
\draw [<-,black] (l) -- (dfo);
\end{tikzpicture}
    \caption{Syntactic relation between the fragments of $\LPH$ (left). Expressive relations between different fragments of  \LPH, LTL and \textit{DFO} (right). A fragment $A$ is strictly more expressive from a fragment $B$ if they are connected via $A\rightarrow B$. Double lined edges denote equal expressive power. }
    \label{fig:exprel}
\end{figure}
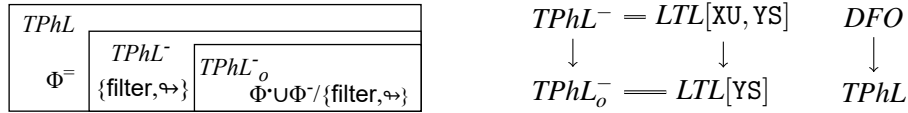

\textbf{\textit{LTL}.} The formulae of \textit{LTL}$[\texttt{X}\texttt{U},\texttt{Y}\texttt{S}]$, given a set of propositions $\mathcal{P}$ are defined as follows:
\begin{align*}
        \phi ::=\  \perp\ |\ p\ |\ \neg \phi_1\ |\ \phi_1 \wedge \phi_2\ |\ \texttt{X} \phi_1\ |\ \phi_1\ \texttt{U}\ \phi_2\ |\ \texttt{Y} \phi_1\ |\ \phi_1\ \texttt{S}\ \phi_2;
\end{align*}
where $\texttt{X}$, and $\texttt{U}$ stand for the next and until modalities, while $\texttt{Y}$, and $\texttt{S}$ stand for previous and since.  The formulae of \textit{LTL}$[\texttt{X}\texttt{U},\texttt{Y}\texttt{S}]$ are interpreted over a discrete, linear model of time, formally represented as $\mathcal{M}_{TL}=\langle T, <, V^{TL}\rangle$, where $T$ is equal to $\mathbb{N}$, $<$ is the linear order and $V^{TL}:\mathcal{P} \rightarrow 2^{T}$ is the interpretation function, mapping each proposition to a set of time instants. The satisfaction relation, i.e., that a formula $\phi$ is true at $t$, is defined as $\mathcal{M}_{TL},t\models \phi$. The semantics of \textit{LTL} are defined as usual; more specifically in what follows we assume the reflexive\footnote{As we work with discrete linear orders, this choice makes no difference.} semantics of $\texttt{S}$ and $\texttt{U}$. 
We denote the pure past fragment of \textit{LTL} i.e., \textit{LTL} without $\texttt{X}$ and $\texttt{U}$ as \textit{LTL}[$\texttt{YS}$].

\textbf{FOMLO.} Given a countable set of variables ${x,y,z,...}$, the formulae of \textit{FOMLO} over a set of unary predicate symbols $\Sigma$ are defined a follows:
\begin{align*}
    \mathit{atomic}::= x < y\ |\ x = y\ |\ P(x)\ (\text{where}\ P \in \Sigma)\\
    \phi ::= \mathit{atomic}\ |\ \neg \phi_1\ |\ \phi_1 \vee \phi_2\ |\ \phi_1 \wedge \phi_2\ |\ \exists x.\ \phi_1\ |\ \forall x.\ \phi_1
\end{align*}
We interpret \textit{FOMLO} formulae over structures of the form $\mathcal{M}_{FO}\langle T,<,V^{FO}\rangle$, where $T$ is equal to $\mathbb{N}$, `$<$' is the linear order while $V^{FO}: \Sigma \rightarrow 2^{T}$ is the interpretation of $\Sigma$. $\mathcal {M}_{FO},t_1,t_2,\cdots,t_n \models \phi(x_1,x_2, \cdots, x_n)$ denotes the satisfaction of a formula $\phi$ with free variables $x_1,x_2, \cdots, x_n$ when they are interpreted as elements $t_i$ of $\mathcal{M}_{FO}$. The semantics of the formulae are defined as usual (see for example~\cite{Rabinovich_2014}). We also define the \textit{FOMLO}$^-$ fragment of \textit{FOMLO}. Syntactically a formula with one free variable $\phi(x)$, is a formula of the fragment if any bounded variable in the negated normal form of $\phi(x)$ is bounded to be $\leq x$. Semantically, this means that for all models a formula $\phi(x)$ of \textit{FOMLO}$^-$ satisfies:
\begin{align}
\begin{mysplit}
\label{eq:semprop}
\forall t\ \mathcal{M}_{FO},t\  \models \phi(x) \leftrightarrow \mathcal{M}_{FO}[0,t],t\models \phi(x)
\end{mysplit}
\end{align}
where ${M}_{FO}[0,t]$ is a finite model starting from $0$ and ending up to position $t$ inclusive. Intuitively, formulae of \textit{FOMLO}$^-$ can talk only about the past and the present. In what follows, given a set of propositions $\mathcal{P}=\{p_1,...,p_k\}$ and a set of predicate symbols $\Sigma=\{p_1(x), ..., p_k(x)\}$ a FOMLO model $\mathcal{M}_{FO}$ is \textit{faithful} to $\mathcal{M}_{TL}$ iff $\forall_{\geq 1}^{\leq k} i. V^{FO}(p_i(x)) = V^{TL}(p_i)$.

\textbf{DFO.} Finally, on our expressiveness analysis we will also consider \textit{DFO}, which in contrast to \textit{FOMLO}, uses dyadic predicate symbols e.g., $p(x,y)$.

\subsection{\LPH$_o^-$ and Pure Past \textit{LTL}}
In this section we will show that \LPH$_o^-$ is expressively equal to the pure past fragment of \textit{LTL}, i.e., \textit{LTL}$[\texttt{Y}\texttt{S}]$.
The \LPH$_o^-$ fragment is described by the triplet $\langle \mathcal{P}^{es},L^-_o,\Phi \rangle$, where $\mathcal{P}^{es}$ is a set defined by the union of event and state predicate sets ($\mathcal{P}^{e/s}$); $L^-_o$ is a set defined by the union of the set of the logical connectives $\{\wedge,\vee,\neg,\in\}$ and the set of temporal operators $\{\rightarrowtail, \sqcup,\sqcap,\setminus\}$. Formulae \LPH$_o^-$ are evaluated over $\mathcal{M}^-=\langle T, I, <,V^{\bigcdot}, V^-\rangle$ models which are defined in a similar manner to the models presented in Section~\ref{sec:language}.

Given a finite set of event propositions\footnote{In what follows, for simplicity, we will refer to atomic predicates of \LPH\ as propositions.} $\mathcal{P}^e=\{e_1,...,e_k\}$, and a finite set of state propositions $\mathcal{P}^s=\{s_1,...,s_k\}$ we say the \textit{FOMLO}$^-$ model $M_{FO}=\langle T,<,V^{FO}\rangle$ is \textit{faithful} to the model $\mathcal{M}^-=\langle T, I, <,V^{\bigcdot}, V^-\rangle$ of \LPH$_o^-$ if it has the following properties: 
\begin{itemize}
    \item for any proposition $e$ in $\mathcal{P}^e$, $V^{\bigcdot}(e)=V^{FO}(e_t)$, and
    \item for any proposition $s$ in $\mathcal{P}^s$, $V^-(s)=\rho(V^{FO}(s_+),V^{FO}(s_\in),V^{FO}(s_-))$ 
\end{itemize}
where $e_t$ corresponds to the monadic predicate $e_t(x)$ and the triplet ($s_+,s_\in,s_-$) corresponds to the monadic predicates $s_+(x)$, $s_\in(x)$, $s_+(x)$ which are true on instants corresponding to the start, intermediate, and end of an interval respectively at which $s$ is true.  $\rho:2^T\times 2^T\times 2^T\rightarrow 2^I$ is a partial mapping from three set of points to a set of intervals of $I$. Given three sets $S,B$ and $E$, corresponding to starting, intermediate and ending points resp., $\rho$ is defined as follows:
\begin{align*}
\rho(S,B,E)=\
\begin{aligned}
     &  \bigl\{[ts,te]:  ts < te \wedge ts \in S \wedge te \in E \wedge \forall_{>ts}^{<ts}t.  (t \in B \wedge t \not\in S \wedge t \not\in E)\bigr\}\\
     & \cup \bigl\{ [ts,\infty): ts \in S \wedge \nexists_{>ts} te.\ te \in E \wedge \forall_{>ts} t.\ t \in B\bigr\}
\end{aligned}
\end{align*}
For all $i\in I$ where a formula $\phi \in \Phi^-$ is true $\rho$ is bijective (recall that $\phi\in\Phi^-$ formulae always hold on disjoint intervals).  For our expressiveness study, we will use the following theorem.
\begin{theorem}[Gabbay et al. 1980~\cite{Gabbay_1980}]
\label{theorem:gabbay}
For every formula of \textit{FOMLO}$^-$ $\phi(x)$ we can find an \textit{LTL}$[\texttt{Y}\texttt{S}]$ formula $\theta$ such that $\forall t. \mathcal{M}_{TL},t\models \phi \leftrightarrow \mathcal{M}_{FO},t\models \theta$ for all $\mathcal{M}_{FO}$ and their faithful models $\mathcal{M}_{TL}$.
\end{theorem}
\begin{proof}
Dual proof of Theorem 2.2 in~\cite{Gabbay_1980}.
\end{proof}
\begin{theorem} 
\label{theorem:idFOMLO}
For every formula $\phi$ of \LPH$_o^-$, for all \LPH$_o^-$ models $\mathcal{M}^-$ and their  \textit{faithful} \textit{FOMLO}$^-$ models $\mathcal{M}_{FO}$:
\begin{enumerate} 
\item if $\phi \in \Phi^{\bigcdot}$, there exists a formula $\phi(t)$ with one free variable  of \textit{FOMLO}$^-$ such that  $\mathcal{M}^-,t\models\phi$ iff $\mathcal{M}_{FO},t\models\phi(t)$.
\item if $\phi \in \Phi^-$, there exist formulae $\phi^{+/\in/-}(t)$ with one free variable such that $\mathcal{M}^-,[ts,te]\models\phi$ iff  
\begin{align*}
\mathcal{M}_{FO},ts \models \phi^+(ts) \wedge \forall_{>ts}^{<te} t.\ \mathcal{M}_{FO},t \models \phi^\in(t) \wedge \mathcal{M}_{FO},te \models \phi^-(te)
\end{align*}
and, $\mathcal{M}^-,[ts,\infty)\models\phi$ iff $
\mathcal{M}_{FO},ts \models \phi^+(ts) \wedge \forall_{>ts} t.\ \mathcal{M}_{FO},t \models\phi^\in(t)$
\end{enumerate}
\end{theorem}
\begin{proof}
        The proof is straightforward by direct translations. We define the translation $\tau^{\bigcdot}$ from formulae of $\Phi^{\bigcdot}$ to $\textit{FOMLO}^-$ formulae as follows:
    \begin{itemize}
        \item $\tau^{\bigcdot}(e,x)= e(x)$
        \item $\tau^{\bigcdot}(\neg \phi,x)= \neg \tau^{\bigcdot}(\phi,x)$
        \item $\tau^{\bigcdot}(\phi [\wedge,\vee] \psi,x) = \tau^{\bigcdot}(\phi,x) [\wedge,\vee] \tau^{\bigcdot}(\psi,x)$
        \item $\tau^{\bigcdot}(\phi \in \psi,x)= \tau^{\bigcdot}(\phi,x) \wedge (\tau^-_+(\psi,x) \vee \tau^-_\in(\psi,x) \vee \tau^-_\in(\psi,x))$
        (We define $\tau^-_{+/\in/-}$ below.)
    \end{itemize}
    Considering that disjoint intervals can be recreated by their starting points, intermediate and endpoints, we define the $\tau^-_+, \tau^-_\in,\tau^-_-$ translation functions  respectively, from formulae of $\Phi^-$ to \textit{FOMLO}$^-$ as follows. 
    \begin{itemize}
        \item $\tau^-_{\{+,\in,-\}}(s,x)=s_{\{+,\in,-\}}(x)$
        \item $\tau^-_+(\phi \rightarrowtail \psi,x)= \tau^{\bigcdot}(\phi,x) \wedge \forall^{<x} z.\ \bigl[\tau^{\bigcdot}(\phi,z) \rightarrow \exists_{>z}^{<t}z'.\ \tau^{\bigcdot}(\psi \wedge \neg \phi,z') \bigr] $
        \item $\tau^-_\in(\phi \rightarrowtail \psi,x)= \exists^{< x} z.\ \bigl[\tau^-_+(\phi \rightarrowtail \psi,z) \wedge \forall^{\leq x}_{>z} z'.\ \neg\tau^{\bigcdot}(\psi \wedge \neg \phi,z')\bigr]$
        \item $\tau^-_-(\phi \rightarrowtail \psi,x)= \tau^{\bigcdot}(\psi \wedge \neg \phi,x) \wedge \exists^{< x} z.\ \bigl[\tau^-_+(\phi \rightarrowtail \psi,z)\wedge \forall^{\leq x}_{>z} z'.\ \neg\tau^{\bigcdot}(\psi \wedge \neg \phi,z')\bigr]$
        \item $\begin{matrix*}[l]
            \tau^-_+(\phi \sqcup \psi,x) =& \bigl[\tau^-_+(\phi,x) \wedge \neg \tau^-_+(\psi,x)\wedge \neg\tau^-_\in(\psi,x) \wedge \neg\tau^-_-(\psi,x)\bigr] \vee\ \bigl[\tau^-_+(\phi,x) \wedge  \tau^-_+(\psi,x)\bigr]\\
        &\vee\ \bigl[\tau^-_+(\psi,x)\wedge \neg \tau^-_+(\phi,x)\wedge \neg\tau^-_\in(\phi,x) \wedge \neg\tau^-_-(\phi,x)\bigr]
        \end{matrix*}$
       
        \item $\begin{matrix*}[l] \tau^-_\in(\phi \sqcup \psi,x)=& \bigl[\tau^-_+(\phi,x) \vee \tau^-_\in(\psi,x)\bigr] \vee \bigl[\tau^-_+(\psi,x) \vee \tau^-_\in(\phi,x)\bigr] \vee \bigl[\tau^-_+(\phi,x) \wedge \tau^-_-(\psi,x)\bigr] \\ 
        & \vee\ \bigl[\tau^-_-(\phi,x) \wedge  \tau^-_+(\psi,x)\bigr]
                \end{matrix*}$

        \item $\begin{matrix*}[l]
            \tau^-_-(\phi \sqcup \psi,x) =& \bigl[\tau^-_-(\phi,x) \wedge \neg \tau^-_+(\psi,x)\wedge \neg\tau^-_\in(\psi,x) \wedge \neg\tau^-_-(\psi,x)\bigr] \vee\ \bigl[\tau^-_-(\phi,x) \wedge  \tau^-_-(\psi,x)\bigr]\\
        &\vee\ \bigl[\tau^-_-(\psi,x)\wedge \neg \tau^-_+(\phi,x)\wedge \neg\tau^-_\in(\phi,x) \wedge \neg\tau^-_-(\phi,x)\bigr]
        \end{matrix*}$
    \end{itemize}
The remaining translations are similar to the ones already presented and therefore omitted. It is easy to see that the conditions for an instant to be the starting, intermediate or endpoint of a formula of \LPH$^-$ is described by \textit{FOMLO}$^-$ formulae. Consequently, given a formula $\phi$ of  $\Phi^{\bigcdot}$, $\mathcal{M}^-,t\models \phi \leftrightarrow \mathcal{M}_{FO},t \models \tau^{\bigcdot}(\phi,t)$. Given a formula $\phi$ of $\Phi^-$ it holds:
\[
\mathcal{M}^-,[ts,te] \models \phi\ \text{iff}\ \mathcal{M}_{FO},ts \models \tau^-_+(\phi,ts) \wedge \forall_{>ts}^{<te} t.\ \mathcal{M}_{FO},t \models\tau^-_\in(\phi,t) \wedge \mathcal{M}_{FO},te \models \tau^-_\in(\phi,te)
\]
Finally, given a formula $\phi$ of $\Phi^-$ it holds:
\[\mathcal{M}^-,[ts,\infty) \models \phi\ \text{iff}\ \mathcal{M}_{FO},ts \models \tau^-_+(\phi,ts) \wedge \forall_{>ts} t.\ \mathcal{M}_{FO},t \models\tau^-_\in(\phi,t)\qedhere\]
\end{proof}
From, Theorems~\ref{theorem:gabbay} and \ref{theorem:idFOMLO} we deduct that:
\begin{theorem}
\label{theorem:idppltl}
For every formula $\phi$ of \LPH$_o^-$, for all models $\mathcal{M}^-$ and their faithful models $\mathcal{M}_{TL}$\footnote{We omit the definition of faithful models of \LPH\ and \textit{LTL} as they are defined in a similar manner to \LPH\ and \textit{FOMLO}.}:
\begin{itemize} 
\item if $\phi \in \Phi^{\bigcdot}$ then there exists a formula $\phi_t$ of \textit{LTL}$[\texttt{Y}\texttt{S}]$ such that  $\mathcal{M}^-,t\models\phi$ iff $\mathcal{M}_{TL},t\models\phi_t$.
\item if $\phi \in \Phi^-$ then there exist formulae $\phi_t^+,\phi_t^\in,\phi_t^-$ of \textit{LTL}$[\texttt{Y}\texttt{S}]$ such that $\mathcal{M}^-,[ts,te]\models\phi$ iff: 
\[\mathcal{M}_{TL},ts \models \phi^+_t \wedge \forall_{>ts}^{<te} t.\ \mathcal{M}_{TL},t \models\phi^\in_t \wedge \mathcal{M}_{TL},te \models \phi^-_t
\]
and, $\mathcal{M}^-,[ts,\infty)\models\phi$ iff $\mathcal{M}_{TL},ts \models \phi^+_t \wedge \forall_{>ts} t.\ \mathcal{M}_{TL},t \models\phi^\in_t$
\end{itemize}
\end{theorem}
Now we will show that \textit{LTL}$[\texttt{Y}\texttt{S}]$ is expressible in \LPH$^-_o$. 
We define the translation $\tau^{r}:\Phi_t\rightarrow\Phi^{\bigcdot}$ where $\Phi_t$ is the set of formulae of \textit{LTL}$[\texttt{Y}\texttt{S}]$ and $\Phi^{\bigcdot}$ is a subset of \LPH$^-_o$ formulae, as follows:
\begin{itemize}
    \item $\tau^r(p)=p$
    \item $\tau^r(\neg \phi) = \neg \tau^r(\phi)$
    \item $\tau^r(\phi\ [\wedge, \vee]\ \psi) = \tau^r(\phi)\ [\wedge, \vee]\ \tau^r(\psi)$
    \item $\tau^r(\texttt{Y}\phi) = \bigl(\tau^r(\phi) \vee \neg \tau^r(\phi)\bigr) \in \bigl( \tau^r(\phi) \rightarrowtail \neg \tau^r(\phi)\bigr) \wedge \neg \plstart\bigl(\tau^r(\phi) \rightarrowtail \neg \tau^r(\phi)\bigr)$
    \item $\tau^r(\phi\ \texttt{S}\ \psi) = \bigl(\tau^r(\phi) \vee \neg \tau^r(\phi)\bigr) \in \bigl( \tau^r(\psi) \rightarrowtail \neg \tau^r(\phi)\bigr)$
\end{itemize}
Note that in the case of $\texttt{Y}$ and $\texttt{S}$, $\tau^r(\phi) \vee \neg \tau^r(\phi)$ is true everywhere but is restricted via the $\in$ modality. It is clear that, given a finite set of propositions $\mathcal{P}$ of \textit{LTL}$[\texttt{Y}\texttt{S}]$, for all models $\mathcal{M}_{TL}$ and their corresponding \textit{faithful} \LPH$^-_o$ models (for all $p \in \mathcal{P}$ it holds $V^{TL}(p) = V^{\bigcdot}(p_e)$ where $p_e$ is an event proposition), and for all \textit{LTL}$[\texttt{Y}\texttt{S}]$ formulae $\phi$ it holds $\mathcal{M}_{TL},t\models \phi $ iff $ \mathcal{M}^-,t\models \tau^r(\phi)$.

\subsection{\LPH$^-$ and \textit{LTL}$[\texttt{X}\texttt{U},\texttt{Y}\texttt{S}]$}
In this section we will show that \LPH$^-$, the extension of \LPH$^-_o$ (i.e., with $\plfilter$ and $\looparrowright$), has equal expressive power with \textit{LTL}$[\texttt{X}\texttt{U},\texttt{Y}\texttt{S}]$. Our approach is similar to the previous section, however this time we will translate formulae of \LPH$^-$ to \textit{FOMLO} (instead of \textit{FOMLO}$^-$). We will use Kamp's theorem:
\begin{theorem}[Kamp~\cite{Kamp1968}]
\label{theorem:kamp}
    Given any \textit{FOMLO} formula $\phi(x)$ with one free variable, there is an LTL formula $\theta$, such that $\theta\equiv\phi(x)$ for all models $\mathcal{M}_{FO}$ and $\mathcal{M}_{TL}$.
\end{theorem}
Consequently, the only thing that remains to prove that \LPH$^-$ is expressible in \textit{LTL}$[\texttt{X}\texttt{U},\texttt{Y}\texttt{S}]$ is to show that formulae involving the minimal range operator ($\looparrowright$) and filtering ($\plfilter$) are expressible in \textit{FOMLO}. Similar to the proof of Theorem~\ref{theorem:idFOMLO}, this is straightforward by extending the translation functions $\tau^-_+, \tau^-_\in,\tau^-_-$ for supporting `$\looparrowright$' and `$\plfilter$'.
We begin with the translation of formulae that involve the minimal range operator ($\looparrowright$):
\begin{align*}
&\tau^-_+(\phi \looparrowright \psi,x) = \tau^{\bigcdot}(\phi,x) \wedge \exists_{>x} te.\ \Bigl[\tau^{\bigcdot}(\psi \wedge \neg \phi, te) \wedge \forall_{>x}^{<te} t.\  \bigl[\neg \tau^{\bigcdot}(\phi,x) \wedge \neg\tau^{\bigcdot}(\psi \wedge \neg \phi, t) \bigr]\Bigr]\\
&\tau^-_\in(\phi \looparrowright \psi,x) =\exists_{<x} ts.\exists_{>x} te.\ \Bigl[\tau^{\bigcdot}(\phi,ts) \wedge \tau^{\bigcdot}(\psi \wedge \neg \phi, te) \wedge \forall_{>ts}^{<te} t.\  \bigl[\neg \tau^{\bigcdot}(\phi,t) \wedge \neg\tau^{\bigcdot}(\psi \wedge \neg \phi, t) \bigr]\Bigr]\\
&\tau^-_-(\phi \looparrowright \psi,x) =\exists_{<x} ts.\Bigl[\tau^{\bigcdot}(\phi,ts)\wedge  \tau^{\bigcdot}(\psi \wedge \neg \phi, x) \wedge \forall_{>ts}^{<x} t.\  \bigl[\neg \tau^{\bigcdot}(\phi,t) \wedge \neg\tau^{\bigcdot}(\psi \wedge \neg \phi, t) \bigr]\Bigr]
\end{align*}
Essentially, the translation of $\phi\ \looparrowright\ \psi$ is similar to the translation of $\phi\ \rightarrowtail \psi$, however in this case it is clear that there is a need for future FOMLO formulae. Concerning the translation of formulae involving filtering ($\plfilter$), first we define formulae $C^{k+}(x_0,\phi,\psi)$ as follows:
\begin{align*}
    C^{k+}(x_0,\phi,\psi)=&
    \exists_{>x_0} x_1. \cdots \exists_{>x_{i-1}} x_i.  \cdots \exists_{>x_{k-1}} x_k. \nexists_{>x_0}^{<x_1} x_{0,1}.\nexists_{>x_i}^{<x_{i+1}} x_{i,i+1}. \cdots \nexists_{>x_{k-1}}^{<x_k} x_{k-1,k}.\\
    &\bigl[ \phi(x_1) \wedge \cdots \wedge \phi(x_{k-1}) \wedge \psi(x_{k})\bigr]
\end{align*}
denoting that $\phi$ is true from $x_1$ to $x_{k-1}$, $\psi$ is true at $x_k$, and all $x_i$ are contiguous and right of $x_0$. Similar to  $C^{k+}$ define the $C^{k-}$ for the left direction from $x_0$. Here, we will only define the translations for the $\plfilter_{<}$ case as the remaining cases can be easily defined in a similar manner.
\begin{align*}
&\tau^-_+(\phi\ \plfilter_{<}\ n ,x) = \tau^-_+(\phi,x) \wedge \bigl[ C^{1+}(x,\tau^-_\in(\phi),\tau^-_-(\phi)) \vee \cdots \vee C^{n-1+}(x,\tau^-_\in(\phi),\tau^-_-(\phi)) \bigl]\\
&\tau^-_\in(\phi\ \plfilter_{<}\ n ,x) =  \tau^-_\in(\phi,x) \wedge \Bigl[\bigl[ C^{1-}(x,\tau^-_\in(\phi),\tau^-_+(\phi)) \wedge C^{n-2+}(x,\tau^-_\in(\phi),\tau^-_-(\phi))\bigl]\\
&\qquad\qquad\qquad\qquad\vee \bigl[C^{2-}(x,\tau^-_\in(\phi),\tau^-_+(\phi)) \wedge C^{n-3+}(x,\tau^-_\in(\phi),\tau^-_-(\phi))\bigr]\vee \cdots\\
&\qquad\qquad\qquad\qquad \vee \bigl[C^{n-2-}(x,\tau^-_\in(\phi),\tau^-_+(\phi)) \wedge C^{1+}(x,\tau^-_\in(\phi),\tau^-_-(\phi))\bigr]\Bigr]\\
& \tau^-_-(\phi\ \plfilter_{<}\ n ,x) = \tau^-_-(\phi,x) \wedge \bigl[ C^{1-}(x,\tau^-_\in(\phi),\tau^-_+(\phi)) \vee \cdots \vee C^{n-1-}(x,\tau^-_\in(\phi),\tau^-_+(\phi)) \bigl]
\end{align*}
It can be seen that although a translation of $\phi\ \plfilter_< n$ exists, the size of the translated formula is linear to $n$. Note that a translation with smaller size might be possible, however for our expressiveness study it is not required to find the optimal translation.

Given all of the above, from Theorem~\ref{theorem:kamp}, it is clear that the analog of Theorem~\ref{theorem:idppltl} also holds for \LPH$^-$ and \textit{LTL}$[\texttt{X}\texttt{U},\texttt{Y}\texttt{S}]$. For the opposite direction, it suffices to show that there are $\tau^r$ translations from formulae involving the remaining \textit{LTL} modalities, i.e., $\texttt{X}$ and $\texttt{U}$, to \LPH$^-$ formulae. For convenience we first define 
$c(\phi)= \phi \rightarrowtail \neg \phi$
where $\phi \in \Phi^{\bigcdot}$. Essentially,  $c(\phi)$ holds for the maximal intervals $[ts,te]$ or $[ts,\infty)$ for which $\forall_{\geq ts}^{<te} t. \mathcal{M},t\models\phi$ or $\forall_{\geq ts} t. \mathcal{M},t\models\phi$ respectively. Therefore we define the corresponding $\tau^r$ translations as follows. 
\begin{align*}
     &\tau^r(\texttt{X}\phi) = (\tau^r(\phi) \in c(\tau^r(\phi))) \wedge \neg \plend(c(\tau^r(\phi))))\ \vee\ \plstart(\neg \tau^r(\phi) \looparrowright\plend(c(\neg \tau^r(\phi))))\\
    &\tau^r(\phi\ \texttt{U}\ \psi) = ((\tau^r(\phi) \vee \neg \tau^r(\phi)) \in ( \plstart(c(\tau^r(\phi)) \sqcap c(\neg \tau^r(\psi))) \looparrowright \\
    &\qquad\qquad\qquad (\plend(c(\tau^r(\phi)) \sqcap c(\neg \tau^r(\psi))) \wedge \tau^r(\psi)))) \vee\ \tau^r(\psi)
\end{align*}
Concerning the translation of $\texttt{X}\phi$, the left part of the disjunction holds true for all instants included in an interval $[ts,te)$ where $c(\tau^r(\phi)))$ is true, while the right part is true at the start of an interval $[t,t+1]$ where $\neg \tau^r(\phi) \looparrowright\plend(c(\neg \tau^r(\phi)))$ is true. In the case of $\tau^r(\phi\ \texttt{U}\ \psi)$, the translation can be divided into two parts: the first part uses the inclusion operator between the tautology $\tau^r(\phi) \vee \neg \tau^r(\phi)$ (true everywhere) and the minimal range formula between (a) the start of a period at which both $\phi$ and $\neg \psi$ are true for all points (excluding the end) and (b) the end of a period $[ts,te)$ at which both $\phi$ and $\neg \psi$ are true and $\psi$ holds at $te$, thus capturing the cases where $\phi$ is true before $\psi$ becomes true; the second part of the translation is the case of $\tau^r(\psi)$ which captures single instances of $\psi$.  
Considering all of the above, we can now say that the \LPH$^-$ fragment of \LPH\ has equal expressive power with \textit{LTL}$[\texttt{X}\texttt{U},\texttt{Y}\texttt{S}]$.

\subsection{Expressiveness of \LPH}
Concerning the complete language \textit{\LPH}, a comparison with \textit{LTL} is not possible as the structures on which semantics is based are incomparable, even for atomic entities. This is because dynamic temporal phenomena may hold on non-disjoint intervals which by default require the half plane of a 2-dimensional temporal space for their representation. 
When compared to \textit{DFO}, it can be easily seen that for all formulae of \LPH\ there are equivalent formulae with two free variables of \textit{DFO}. For example, a dynamic temporal phenomenon proposition $p$ can be represented by a dyadic predicate $p(x,y)$ such that $\mathcal{M},[ts,te]\models p\leftrightarrow\mathcal{M}_{DFO},ts,te \models p(x,y)$. In a similar manner to the previous sections, we can define translations from  \textit{\LPH} to \textit{DFO} for the remaining formulae---this time however with two free variables corresponding to starting and ending instants. However, the reverse direction does not hold. This can be shown with the following  example. Consider the \textit{DFO} formula $\phi(x,y)=\neg p(x,y)$, since negation is not included for formulae\footnote{We chose to omit negation from formulae of $\Phi^=$ as it would affect significantly the performance of our implementation.} of $\Phi^=$ there is no formula $\phi_t$ of \LPH\ such that $\mathcal{M},[ts,te] \models \phi_t \leftrightarrow \mathcal{M}_{DFO},ts,te\models\neg p(x,y)$ for all models. Therefore $\LPH$ is strictly less expressive than \textit{DFO}.       
\section{Stream processing}
\label{sec:stream}
In this section, we formally present the correctness criteria for stream processing with the \LPH\ language. Given a stream, i.e., an arbitrary long sequence of time associated atomic formulae of $\Phi$, the evaluation at a given instant $t$, of formulae that refer only to the past ($\phi$ of \LPH$^-_o$) is an easy task, as their truth value can be determined for all $t' \leq t$ (see Equation~\eqref{eq:semprop}). However, this is not the case for formulae such as $\tau^r(\texttt{X}p)$ and $\tau^r(p\ \texttt{U}\ q$) that refer to the future, as their truth value at an instant $t$ may depend on future information ($>t$). Consequently, in order to guarantee \textit{correctness},  \textit{monotonicity} and \textit{punctuality}---we will define these shortly---, the two valued semantics of \LPH, are not sufficient for the evaluation of formulae on constantly evolving streams. In order to treat the issue of evaluations with unknown status---i.e., when all required information is not available at current time---, we follow an approach similar to~\cite{Bauer_2011}. We extend the semantics of \LPH\ for stream processing, to utilise three values: true ($\top$), false ($\bot$) and unknown ($?$). 
Due to space limitations, we will omit the presentation of the three valued semantics\footnote{The complete three valued semantics are available in~\url{https://manospits.github.io/files/Three_valued_semantics.pdf}.} in this paper, instead we will focus on formalising the notions of stream, stream processor, and define the properties of  correctness, punctuality and monotonicity for stream processors of $\LPH$. 

A \textit{stream}, at any instant $t$ can be represented by the finite model $\mathcal{M}_t=\langle T_t, I_t, <,V^{\bigcdot}, V^-, V^= \rangle$ where  $T_t=\{0,1,\cdots,t\}$, $I_t= T_t\times T_t\cup\{[ts,\infty):ts\in T_t\}$, and $V^{\bigcdot}, V^-, V^=$ are valuation functions defined in similar manner to Section~\ref{sec:language}. 

A \textit{stream processor}, in symbols $\mathcal{SP}_t$, is defined by the triplet $\langle \Lambda^{\bigcdot}_t,\Lambda^-_t,\Lambda^=_t \rangle$ where $t\in T$,
 $\Lambda^{\bigcdot}_t:\Phi^{\bigcdot}\times T_t \rightarrow \{\top,\bot,?\}$, $\Lambda^-_t:\Phi^{-}\times (I_t^c\cup I^+_T) \rightarrow \{\top,\bot,?\}$, and $\Lambda^=_t:\Phi^{=}\times(I^c_t\cup I^+_t ) \rightarrow \{\top,\bot,?\}$, are formulae valuation functions assigning truth values on formulae-instants/intervals pairs and $I^c_t=T_t\times T_t$ and  $I^+_t=\{[ts,t+] : ts,t \in T_t\}$. Intervals of $I_t^c$, (e.g., [ts,te]) denote that a phenomenon started at $ts$ and ended at $te$, while intervals of $I^+_t$,  (e.g., $[ts,t+]$) denote that a phenomenon started at $ts$, and continues to be true/unknown at $t$ but does not end at $t$. Intervals of $I^+_t$ are useful for capturing the truth value of valuations that are true but are still ongoing---see for example the semantics of $\phi \rightarrowtail \psi$ for intervals open to $\infty$. We assume that the input phenomena are ordered and their truth value is never unknown. Now, we define the \textit{correctness}, \textit{punctuality} and \textit{monotonicity} properties for $\mathcal{SP}$.
 
 \textbf{Correctness.} A stream processor has the \textit{correctness} property iff given any stream, for all $t$ and for any $\phi\in \Phi$ evaluation by $\mathcal{SP}_t$ (i.e., via $\Lambda^{\bigcdot/-/=}$) that is true (false) at an instant $t_i$  or interval $i$, $\phi$ is also true (false) (i.e., via the semantics of Section~\ref{sec:language}) at $t_i$ or $i$ in $\mathcal{M}_{t}$.

\textbf{Monotonicity.} A stream processor has the \textit{motonocity} property iff given any stream, for all $t$ and for $\phi\in \Phi$ evaluation by $\mathcal{SP}_t$ that is true (false) at an instant $t_i$ or interval $i$, will also be evaluated to be true (false) at $t_i$ or $i$ by all $\mathcal{SP}_{t'}$ with $t'>t$.

 \textbf{Punctuality.} A stream processor has the punctuality property iff given any stream, for any $\phi$, and for all instants $t_i$ or intervals $i$ if there exists minimum $t\geq t_i$ such that $\phi$ is true (false) for all $t'\geq t$ in all $\mathcal{M}_t'$ at $t_i$ or $t$ then $\mathcal{SP}_t$ evaluates $\phi$ to be true (false) at an instant $t_i$ or interval $i$.

We say that a stream processor for \LPH\ is  \textit{proper} iff it has all the three aforementioned properties. It is easy to see that some formulae, given certain streams, can never be true but always stay unknown. For example, consider the formula $\tau^r(\texttt{G}p)$, where $\texttt{G}$ is the `globally' \textit{LTL} operator, and a stream where $p$ is true at all instants; at any given point $t$ in time, the stream processor is agnostic to the future, therefore in order to maintain the \textit{monotonicity} property, $\tau^r(\texttt{G}p)$ will be evaluated by $\mathcal{SP}_t$ for all $t'\leq t$ to be unknown.

Phenesthe+, is a $\textit{proper}$ stream processor of \LPH. While we will not present a formal proof in this paper, we will briefly discuss its processing and its implementation. Phenesthe+ is a complex event processing engine that given an input stream and a set of temporal phenomena definitions, will produce an output stream of temporal phenomena detections i.e., phenomena associated with a set of instants or intervals at which they are true. Compared to automata based methods, the phenomena are compiled via rewriting into an internal Prolog representation which is then later used for processing. This procedure is linear with respect to the size of the formulae involved. The phenomena definitions can be hierarchical, and their processing, if possible, can happen in parallel.  Phenesthe+ detects phenomena by performing temporal queries over tumbling temporal windows of size equal to a user defined step ($ST$) size. A temporal window contains all the new information that has arrived since the last temporal query, as well as information from previous windows that has a possible future use. For example, given the formula $\phi \looparrowright \psi$, if $\phi$ is true in the current window but $\psi$ is not, then $\phi$ must be retained. Note that information from previous windows was also retained in the previous version of Phenesthe for valuations that are true but ongoing, or involved dynamic temporal phenomena. All information outside the temporal window that does not have `future use' is discarded. From a practical perspective it is not viable to keep everything from the past that can contribute to a future detection. Therefore, we allow setting a maximum limit for past information (retaining threshold $RT$). When this threshold is active Phenesthe+ is no longer \textit{proper} with respect to the full stream, but remains \textit{proper} for the part of the stream that is bounded by $RT$. Similarly, the \textit{punctuality} property depends on $ST$, if $ST=1$, then Phenesthe+ is punctual, however if $ST>1$, detections will be produced at the latest $ST-1$ time units after their punctual time. 

In terms of complexity, evaluation of formulae in Phenesthe+ happens via single-scan or in the worst case, i.e., when overlapping intervals are involved, polynomial algorithms with respect to the size of the structure (current temporal window). It has to be noted that while $\LPH^-$ has equal expressive power with LTL, it can accomplish efficient processing of phenomena definitions by utilising intervals to represent set of points. For example an interval $[ts,te]$ produced by the evaluation of the formula $\phi\rightarrowtail\psi$ requires only two points for the representation of all the instants included in $[ts,te]$, therefore contributing significantly to space and processing time economy.

\section{Experimental Evaluation}
\label{sec:evaluation}
We presented the theoretical basis of $\LPH$. Now, we will evaluate the efficiency of our extended stream processing engine on a reproducible\footnote{\url{https://github.com/manospits/Phenesthe/tree/future}} maritime monitoring scenario.

\textbf{Experimental setup}
For our experimental evaluation we use a public dataset containing AIS vessel data,  transmitted over a period of 6 months, from October 1st, 2015 to March 31st, 2016, in the area of Brest, France~\cite{RAY2019} along with spatio-temporal events relating vessels with areas (in total $\approx$ 16M input events). We run our experiments on machine with an Intel i7-3770 CPU running Ubuntu 20.04.6 LTS. The set of maritime phenomena we detect as well as the input events are summarised in Table~\ref{tab:datasum}. We compare stream processing efficiency when the set of maritime phenomena definitions includes and does not include phenomena marked with $\ddag$, i.e., phenomena that utilise $\looparrowright$ and, or $\plfilter$ or depend on phenomena that utilise them. 

\begin{table}[t]
\caption{Input and output phenomena description. `IE', `UE', `US' and `UD' stand for Input/User Event/State/Dynamic temporal phenomenon. Phenomena with  $\dag$ have future dependencies while phenomena with $\ddag$ utilise the new temporal modalities or depend on phenomena that utilise them. The last column lists approximately the number of input or output instants/intervals.}
\label{tab:datasum}
\footnotesize	
\centering
\begin{tabular}{@{}cllc@{}}
\toprule
\textbf{Type}                                    & \textbf{Phenomenon}                             & \textbf{Description}                                                                                                 & \textbf{Number}       \\ \midrule
\multirow{2}{*}{IE} & $\mt{ais}(V,S,C,H)$                           & \begin{tabular}[c]{@{}l@{}}AIS transmitted information  (vessel ID, speed, course, heading)\end{tabular} & 15.8M \\[2pt]
\multicolumn{1}{c}{}                    & $\mt{enters/leaves\{Port,Fishing\}}(V,A)$ & Vessel enters/leaves port/fishing area.                                                                     & 160K  \\
\hline
UE                                     & $\mt{stop\_start/end}(V)$                     & Start/end of a stop.                                                                                        &  800K            \\ \hline
\multirow{9}{*}{US}                     & $\mt{in\_\{port,fishing\}\_area}(V,A)$        & In fishing/port area.                                                                                       & 70K             \\[2pt]
                                        & $\mt{stopped}(V)$                            & Stopped vessel.                                                                                             &      300K        \\[2pt]
                                        & $\mt{underway}(V)$                            & Vessel underway.                                                                                            &    132K          \\[2pt]
                                        & $\mt{moored}(V)$                              & Moored vessel.                                                                                              &        323K      \\[2pt]
                                        & $^{\dag\ddag}\mt{fishing\_warning}(V, F)$                 & \begin{tabular}[c]{@{}l@{}}Warning: Non fishing vessel possibly engaged in fishing.\end{tabular}         &            7K  \\[2pt]
                                        & $^{\dag\ddag}\mt{waiting\_time}(V, P)$                    & Port waiting time.                                                                                          &         42K     \\[2pt]
                                        & $^{\dag\ddag}\mt{long\_waiting\_time}$(V, P)              & \begin{tabular}[c]{@{}l@{}}Warning: waiting time longer than a threshold.\end{tabular}                   &       28K       \\[2pt]
                                        & $\mt{unusual\_stop}(V)$                       & \begin{tabular}[c]{@{}l@{}}Warning: vessel performs a stop in an unexpected area.\end{tabular}               &     27K         \\[2pt]
                                        & $^{\dag\ddag}\mt{possible\_malfunction}(V)$               & \begin{tabular}[c]{@{}l@{}}Warning: Vessel might have a malfunction.\end{tabular}                         & 3K             \\\hline
\multirow{3}{*}{UD}                     & $^\dag\mt{trip}(V, PA, PB)$                        & Vessel trip from PA to PB.                                                                                  &   39K           \\[2pt]
                                        & $^{\dag\ddag}\mt{suspicious\_trip}(V, PA, PB)$            & Trip from PA to PB contained warnings.                                                                      &   3K           \\[2pt]
                                        & $^\dag\mt{fishing\_trip}(V, PA, FA, PB)$           & \begin{tabular}[c]{@{}l@{}}Fishing trip from $PA$ to $PB$ contained fishing in $FA$.\end{tabular}               &    6K          \\ \bottomrule
\end{tabular}
\end{table}

\textbf{Experimental results}
The results of our evaluation are illustrated in Figure~\ref{fig:results}. We perform complex event processing with $ST=3h$, and  $RT=\{2,4,8,16\}$ days. Figure~\ref{fig:results} (left) shows the average processing time for each experiment. The results show that the addition of $\ddag$ phenomena does not affect processing efficiency significantly, but also that Phenesthe+ is capable of producing detections in less than 2 seconds (multithreaded) when the data retaining threshold is set to 16 days. Note that the performance gain by running the multithreaded version of Phenesthe+ depends on the dependecies between phenomena and the computation of the processing order. For example, $\mt{fishing\_warning}$ and $\mt{waiting\_time}$ can be processed in parallel while $\mt{fishing\_warning}$ and $\mt{suspicious\_trip}$ cannot. In~\cite{pitsikalis22a} we describe the computation of the processing order. With the addition of the new temporal phenomena, we limited the number of phenomena that can be processed in parallel. The results of Figure~\ref{fig:results} (left) confirm this. We also perform complex event processing with $ST=24h$, and set $RT=\infty$ (i.e., keep non-redundant information forever)---recall that our dataset involves a 6 month period. Similar to the previous experiments we compare performance when Phenesthe+ is executed in parallel or serial manner, and with or without temporal phenomena marked with $\ddag$ (see Table~\ref{tab:datasum}). Figure~\ref{fig:results} (middle) shows the average processing time  while Figure~\ref{fig:results} (right) shows the average number of input entities plus retained instants/ intervals for each case. The results show, that in terms of processing time, performance is not significantly affected when including $\ddag$ phenomena in both serial and parallel processing even when Phenesthe+ retains all information. In more detail, apart from the input events (on average 90K per temporal query) when we include $\ddag$ phenomena the number of retained instants/intervals increases on average by 20K, therefore bringing the total number of input+retained entities up to 120K. Even in this setting Phenesthe+, produces detections in approximately 2  and 4  seconds (serial and parallel respectively).
\begin{figure}[t]
\centering
\includegraphics[width=0.3\textwidth]{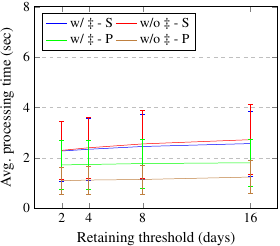}
\includegraphics[width=0.3\textwidth]{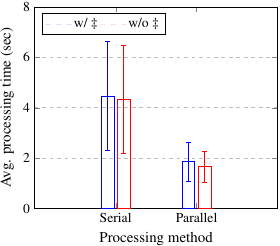}
\includegraphics[width=0.333\textwidth]{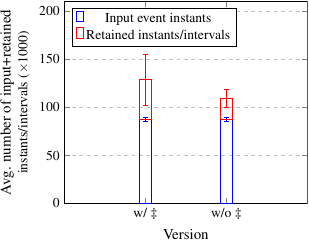}
\caption{Experimental results. Average processing time per query  with $ST=3h$ and $RT$ = $\{2,4,8,16\}$ days or $ST=24h$ and $RT$ = $\infty$ (left) and (middle) respectively. Average number of input plus retained instants/ intervals per query when $ST=24h$ and $RT=\infty$ (right).}
\label{fig:results}
\end{figure}
 
\section{Related work \& Discussion}
\label{sec:relateddisc}
There are several very expressive temporal logics. The HS logic~\cite{halpern1991} is a very powerful logic for representing both instantaneous and durative temporal phenomena. When time is linear and the intervals homogeneous (therefore non overlapping) the HS logic is equally expressive with \textit{LTL} but is exponentially more succinct~\cite{Bozzeli_2018}. In its original version, the HS logic does not make any assumptions on the nature of intervals. In this paper, we showed the \LPH$^-$ has equal expressive power to \textit{LTL}, therefore concerning the linear HS variant studied in~\cite{Bozzeli_2018}, \LPH$^-$ is equally expressive. It is well known, that the chop operator of Venema's CDT logic~\cite{Venema_1990} is inexpressible in the HS logic~\cite{Lodaya_2000,Bresolin_2008}.  \LPH\ supports the chop operator in the form of $\plmeets$. While a formal expressiveness comparison of \LPH\ with the HS or CDT would be desirable, the omission of negation from $\Phi^=$ formulae makes this a challenging task.

Concerning our criteria for proper stream processors of \LPH, as mentioned earlier, their concepts are not entirely new. For example ``correctness'' has a similar notion with ``soundness'' in run-time verification~\cite{Bartocci2018} (i.e., the output should be correct with respect to the specification). Likewise, the ``monotonicity'' property as we have defined it, in run-time monitoring appears as the irrevocability property respectively for monitors~\cite{Aceto2019}---a monitor that has the ``irrevocability'' property is unable to revoke the acceptance or the rejection of a trace. Finally, the ``punctuality'' property can be related to the ``tightness'' property of monitors~\cite{Aceto2019}, under which monitors are restricted to make a choice as soon as there is sufficient information available. In this work, we utilise the similar notions from run-time monitoring and verification for the task of complex event processing. 

From a complex event processing perspective, LARS is a logic based framework for reasoning over streams~\cite{Beck_2018}. While the language of LARS is expressible in \textit{LTL}, the reverse direction does not hold, since LARS does not support `until'. A well known runtime monitoring system with point-based semantics is LOLA~\cite{Angelo2005}. While LOLA does not allow durative phenomena, we saw in section~\ref{sec:expressiveness} that formulae that hold on disjoint intervals can be expressed using point based modalities. It is not possible, however, to model formulae that hold on overlapping intervals. Furthermore, in the worst case LOLA requires memory equal to the size of the trace so far, which is not practical for large industrial applications such as maritime monitoring. In Phenesthe+ we allow the user to choose the retaining threshold.
 A complex event recognition framework is RTEC~\cite{Artikis_2015}. RTEC is a logic based formalism whereby events and fluents are expressed with a variant of the Event Calculus~\cite{kowalski86}. While there isn't a formal study of the expressive power of the language of RTEC, its semantics suggest that it has at most equal expressive power with pure past \textit{LTL}. 
Bauer et. al.~\cite{Bauer_2011}, propose three valued semantics for monitoring \textit{LTL} formulae. In our work, we also use three valued semantics, however for a more general case, as our language allows the representation of temporal phenomena that hold on overlapping intervals, which cannot be modeled in \textit{LTL}. Team semantics for LTL~\cite{Krebs18} or HyperLTL~\cite{Clarkson2014} offer a promising direction towards the representation of concurrent temporal phenomena, however they are limited to a finite number of concurrent traces. In $\LPH$\ a dynamic temporal phenomenon may hold on possibly infinite overlapping intervals.
 
Closing, in this paper we presented \LPH\ and studied the expressive power of its different fragments. Specifically, we showed that \LPH$^-_o$ has equal expressive power with pure past \textit{LTL} while its extension, \LPH$^-$, has equal expressive power with \textit{LTL}. Concerning the complete logic \LPH, we showed that it is strictly less expressive than dyadic first-order logic. Moreover, we defined criteria for \textit{proper} stream processors that use our language, and evaluated Phenesthe+, our stream processing implementation on real maritime data. Our results, show that Phenesthe+ is suitable for the task of maritime monitoring as it produces results in real-time. While the application of our experiment involved the maritime domain, Phenesthe+ is generic, and can be applied in other areas. 

Regarding future work, we aim to study the expressive power of a theoretical variant of \LPH\ that includes negation on formulae of $\Phi^=$ in comparison with two-dimensional modal logics. Furthermore, as one of the main motivations for the creation of \LPH\ was facilitating writing temporal formulae, we plan to compare succinctness of \LPH$^-$ formulae with \textit{LTL} formulae. Finally, we aim to apply Phenesthe+ for human activity monitoring in smart homes.
\bibliographystyle{eptcs}
\bibliography{bibliography}

\begin{thebibliography}{10}
\providecommand{\bibitemdeclare}[2]{}
\providecommand{\surnamestart}{}
\providecommand{\surnameend}{}
\providecommand{\urlprefix}{Available at }
\providecommand{\url}[1]{\texttt{#1}}
\providecommand{\href}[2]{\texttt{#2}}
\providecommand{\urlalt}[2]{\href{#1}{#2}}
\providecommand{\doi}[1]{doi:\urlalt{https://doi.org/#1}{#1}}
\providecommand{\eprint}[1]{arXiv:\urlalt{https://arxiv.org/abs/#1}{#1}}
\providecommand{\bibinfo}[2]{#2}

\bibitemdeclare{article}{Aceto2019}
\bibitem{Aceto2019}
\bibinfo{author}{Luca \surnamestart Aceto\surnameend}, \bibinfo{author}{Antonis
  \surnamestart Achilleos\surnameend}, \bibinfo{author}{Adrian \surnamestart
  Francalanza\surnameend}, \bibinfo{author}{Anna \surnamestart
  Ing\'{o}lfsd\'{o}ttir\surnameend} \& \bibinfo{author}{Karoliina \surnamestart
  Lehtinen\surnameend} (\bibinfo{year}{2019}): \emph{\bibinfo{title}{Adventures
  in Monitorability: From Branching to Linear Time and Back Again}}.
\newblock {\slshape \bibinfo{journal}{Proc. ACM Program. Lang.}}
  \bibinfo{volume}{3}(\bibinfo{number}{POPL}), \doi{10.1145/3290365}.

\bibitemdeclare{article}{Allen_1983}
\bibitem{Allen_1983}
\bibinfo{author}{James~F. \surnamestart Allen\surnameend}
  (\bibinfo{year}{1983}): \emph{\bibinfo{title}{Maintaining knowledge about
  temporal intervals}}.
\newblock {\slshape \bibinfo{journal}{Communications of the ACM}}
  \bibinfo{volume}{26}(\bibinfo{number}{11}), p. \bibinfo{pages}{832–843},
  \doi{10.1145/182.358434}.

\bibitemdeclare{inbook}{Anicic_2010}
\bibitem{Anicic_2010}
\bibinfo{author}{Darko \surnamestart Anicic\surnameend}, \bibinfo{author}{Paul
  \surnamestart Fodor\surnameend}, \bibinfo{author}{Sebastian \surnamestart
  Rudolph\surnameend}, \bibinfo{author}{Roland \surnamestart
  Stühmer\surnameend}, \bibinfo{author}{Nenad \surnamestart
  Stojanovic\surnameend} \& \bibinfo{author}{Rudi \surnamestart
  Studer\surnameend} (\bibinfo{year}{2010}): \emph{\bibinfo{title}{A Rule-Based
  Language for Complex Event Processing and Reasoning}}, p.
  \bibinfo{pages}{42–57}.
\newblock {\slshape \bibinfo{series}{Lecture Notes in Computer Science}}
  \bibinfo{volume}{6333}, \bibinfo{publisher}{Springer Berlin Heidelberg},
  \doi{10.1007/978-3-642-15918-3\_5}.

\bibitemdeclare{article}{Artikis_2015}
\bibitem{Artikis_2015}
\bibinfo{author}{Alexander \surnamestart Artikis\surnameend},
  \bibinfo{author}{Marek \surnamestart Sergot\surnameend} \&
  \bibinfo{author}{Georgios \surnamestart Paliouras\surnameend}
  (\bibinfo{year}{2015}): \emph{\bibinfo{title}{An Event Calculus for Event
  Recognition}}.
\newblock {\slshape \bibinfo{journal}{IEEE Transactions on Knowledge and Data
  Engineering}} \bibinfo{volume}{27}(\bibinfo{number}{4}), pp.
  \bibinfo{pages}{895--908}, \doi{10.1109/TKDE.2014.2356476}.

\bibitemdeclare{book}{Bartocci2018}
\bibitem{Bartocci2018}
\bibinfo{editor}{Ezio \surnamestart Bartocci\surnameend} \&
  \bibinfo{editor}{Yli{\`{e}}s \surnamestart Falcone\surnameend}, editors
  (\bibinfo{year}{2018}): \emph{\bibinfo{title}{Lectures on Runtime
  Verification - Introductory and Advanced Topics}}.
\newblock {\slshape \bibinfo{series}{Lecture Notes in Computer Science}}
  \bibinfo{volume}{10457}, \bibinfo{publisher}{Springer},
  \doi{10.1007/978-3-319-75632-5}.

\bibitemdeclare{article}{Bauer_2011}
\bibitem{Bauer_2011}
\bibinfo{author}{Andreas \surnamestart Bauer\surnameend},
  \bibinfo{author}{Martin \surnamestart Leucker\surnameend} \&
  \bibinfo{author}{Christian \surnamestart Schallhart\surnameend}
  (\bibinfo{year}{2011}): \emph{\bibinfo{title}{Runtime Verification for LTL
  and TLTL}}.
\newblock {\slshape \bibinfo{journal}{ACM Trans. Softw. Eng. Methodol.}}
  \bibinfo{volume}{20}(\bibinfo{number}{4}), \doi{10.1145/2000799.2000800}.

\bibitemdeclare{article}{Beck_2018}
\bibitem{Beck_2018}
\bibinfo{author}{Harald \surnamestart Beck\surnameend}, \bibinfo{author}{Minh
  \surnamestart Dao-Tran\surnameend} \& \bibinfo{author}{Thomas \surnamestart
  Eiter\surnameend} (\bibinfo{year}{2018}): \emph{\bibinfo{title}{LARS: A
  Logic-based framework for Analytic Reasoning over Streams}}.
\newblock {\slshape \bibinfo{journal}{Artificial Intelligence}}
  \bibinfo{volume}{261}, pp. \bibinfo{pages}{16--70},
  \doi{10.1016/j.artint.2018.04.003}.

\bibitemdeclare{inproceedings}{bohlen1998}
\bibitem{bohlen1998}
\bibinfo{author}{Michael~H \surnamestart Bohlen\surnameend},
  \bibinfo{author}{Renato \surnamestart Busatto\surnameend} \&
  \bibinfo{author}{Christian~S. \surnamestart Jensen\surnameend}
  (\bibinfo{year}{1998}): \emph{\bibinfo{title}{Point-versus interval-based
  temporal data models}}.
\newblock In: {\slshape \bibinfo{booktitle}{Proceedings 14th International
  Conference on Data Engineering}}, pp. \bibinfo{pages}{192--200},
  \doi{10.1109/ICDE.1998.655777}.

\bibitemdeclare{article}{Bozzeli_2018}
\bibitem{Bozzeli_2018}
\bibinfo{author}{Laura \surnamestart Bozzelli\surnameend},
  \bibinfo{author}{Alberto \surnamestart Molinari\surnameend},
  \bibinfo{author}{Angelo \surnamestart Montanari\surnameend},
  \bibinfo{author}{Adriano \surnamestart Peron\surnameend} \&
  \bibinfo{author}{Pietro \surnamestart Sala\surnameend}
  (\bibinfo{year}{2018}): \emph{\bibinfo{title}{Interval vs. Point Temporal
  Logic Model Checking: An Expressiveness Comparison}}.
\newblock {\slshape \bibinfo{journal}{ACM Trans. Comput. Logic}}
  \bibinfo{volume}{20}(\bibinfo{number}{1}), \doi{10.1145/3281028}.

\bibitemdeclare{inproceedings}{Bresolin_2008}
\bibitem{Bresolin_2008}
\bibinfo{author}{Davide \surnamestart Bresolin\surnameend},
  \bibinfo{author}{Dario \surnamestart Della~Monica\surnameend},
  \bibinfo{author}{Valentin \surnamestart Goranko\surnameend},
  \bibinfo{author}{Angelo \surnamestart Montanari\surnameend} \&
  \bibinfo{author}{Guido \surnamestart Sciavicco\surnameend}
  (\bibinfo{year}{2008}): \emph{\bibinfo{title}{Decidable and Undecidable
  Fragments of Halpern and Shoham's Interval Temporal Logic: Towards a Complete
  Classification}}.
\newblock In \bibinfo{editor}{Iliano \surnamestart Cervesato\surnameend},
  \bibinfo{editor}{Helmut \surnamestart Veith\surnameend} \&
  \bibinfo{editor}{Andrei \surnamestart Voronkov\surnameend}, editors:
  {\slshape \bibinfo{booktitle}{Logic for Programming, Artificial Intelligence,
  and Reasoning}}, \bibinfo{publisher}{Springer Berlin Heidelberg},
  \bibinfo{address}{Berlin, Heidelberg}, pp. \bibinfo{pages}{590--604},
  \doi{10.1007/978-3-540-89439-1_41}.

\bibitemdeclare{inproceedings}{Clarkson2014}
\bibitem{Clarkson2014}
\bibinfo{author}{Michael~R. \surnamestart Clarkson\surnameend},
  \bibinfo{author}{Bernd \surnamestart Finkbeiner\surnameend},
  \bibinfo{author}{Masoud \surnamestart Koleini\surnameend},
  \bibinfo{author}{Kristopher~K. \surnamestart Micinski\surnameend},
  \bibinfo{author}{Markus~N. \surnamestart Rabe\surnameend} \&
  \bibinfo{author}{C{\'{e}}sar \surnamestart S{\'{a}}nchez\surnameend}
  (\bibinfo{year}{2014}): \emph{\bibinfo{title}{Temporal Logics for
  Hyperproperties}}.
\newblock In \bibinfo{editor}{Mart{\'{\i}}n \surnamestart Abadi\surnameend} \&
  \bibinfo{editor}{Steve \surnamestart Kremer\surnameend}, editors: {\slshape
  \bibinfo{booktitle}{Principles of Security and Trust - Third International
  Conference, {POST} 2014, Held as Part of the European Joint Conferences on
  Theory and Practice of Software, {ETAPS} 2014, Grenoble, France, April 5-13,
  2014, Proceedings}}, {\slshape \bibinfo{series}{Lecture Notes in Computer
  Science}} \bibinfo{volume}{8414}, \bibinfo{publisher}{Springer}, pp.
  \bibinfo{pages}{265--284}, \doi{10.1007/978-3-642-54792-8\_15}.

\bibitemdeclare{inproceedings}{Cugola_Margara_2010}
\bibitem{Cugola_Margara_2010}
\bibinfo{author}{Gianpaolo \surnamestart Cugola\surnameend} \&
  \bibinfo{author}{Alessandro \surnamestart Margara\surnameend}
  (\bibinfo{year}{2010}): \emph{\bibinfo{title}{TESLA: a formally defined event
  specification language}}.
\newblock In: {\slshape \bibinfo{booktitle}{DEBS ’10}},
  \bibinfo{publisher}{ACM Press}, p.~\bibinfo{pages}{50},
  \doi{10.1145/1827418.1827427}.

\bibitemdeclare{inproceedings}{Angelo2005}
\bibitem{Angelo2005}
\bibinfo{author}{Ben \surnamestart D'Angelo\surnameend},
  \bibinfo{author}{Sriram \surnamestart Sankaranarayanan\surnameend},
  \bibinfo{author}{Cesar \surnamestart Sanchez\surnameend},
  \bibinfo{author}{Will \surnamestart Robinson\surnameend},
  \bibinfo{author}{Bernd \surnamestart Finkbeiner\surnameend},
  \bibinfo{author}{Henny~B. \surnamestart Sipma\surnameend},
  \bibinfo{author}{Sandeep \surnamestart Mehrotra\surnameend} \&
  \bibinfo{author}{Zohar \surnamestart Manna\surnameend}
  (\bibinfo{year}{2005}): \emph{\bibinfo{title}{LOLA: Runtime Monitoring of
  Synchronous Systems}}.
\newblock In: {\slshape \bibinfo{booktitle}{Proceedings of the 12th
  International Symposium on Temporal Representation and Reasoning}},
  \bibinfo{series}{TIME '05}, \bibinfo{publisher}{IEEE Computer Society},
  \bibinfo{address}{USA}, p. \bibinfo{pages}{166–174},
  \doi{10.1109/TIME.2005.26}.

\bibitemdeclare{inproceedings}{Dohr2018}
\bibitem{Dohr2018}
\bibinfo{author}{Andreas \surnamestart Dohr\surnameend},
  \bibinfo{author}{Christiane \surnamestart Engels\surnameend} \&
  \bibinfo{author}{Andreas \surnamestart Behrend\surnameend}
  (\bibinfo{year}{2018}): \emph{\bibinfo{title}{{Algebraic Operators for
  Processing Sets of Temporal Intervals in Relational Databases}}}.
\newblock In \bibinfo{editor}{Natasha \surnamestart Alechina\surnameend},
  \bibinfo{editor}{Kjetil \surnamestart N{\o}rv{\aa}g\surnameend} \&
  \bibinfo{editor}{Wojciech \surnamestart Penczek\surnameend}, editors:
  {\slshape \bibinfo{booktitle}{25th International Symposium on Temporal
  Representation and Reasoning (TIME 2018)}}, {\slshape
  \bibinfo{series}{Leibniz International Proceedings in Informatics (LIPIcs)}}
  \bibinfo{volume}{120}, \bibinfo{publisher}{Schloss Dagstuhl--Leibniz-Zentrum
  fuer Informatik}, \bibinfo{address}{Dagstuhl, Germany}, pp.
  \bibinfo{pages}{11:1--11:16}, \doi{10.4230/LIPIcs.TIME.2018.11}.

\bibitemdeclare{inproceedings}{Gabbay_1980}
\bibitem{Gabbay_1980}
\bibinfo{author}{Dov \surnamestart Gabbay\surnameend}, \bibinfo{author}{Amir
  \surnamestart Pnueli\surnameend}, \bibinfo{author}{Saharon \surnamestart
  Shelah\surnameend} \& \bibinfo{author}{Jonathan \surnamestart
  Stavi\surnameend} (\bibinfo{year}{1980}): \emph{\bibinfo{title}{On the
  Temporal Analysis of Fairness}}.
\newblock In: {\slshape \bibinfo{booktitle}{Proceedings of the 7th ACM
  SIGPLAN-SIGACT Symposium on Principles of Programming Languages}},
  \bibinfo{series}{POPL '80}, \bibinfo{publisher}{Association for Computing
  Machinery}, \bibinfo{address}{New York, NY, USA}, p.
  \bibinfo{pages}{163–173}, \doi{10.1145/567446.567462}.

\bibitemdeclare{article}{halpern1991}
\bibitem{halpern1991}
\bibinfo{author}{Joseph~Y. \surnamestart Halpern\surnameend} \&
  \bibinfo{author}{Yoav \surnamestart Shoham\surnameend}
  (\bibinfo{year}{1991}): \emph{\bibinfo{title}{A propositional modal logic of
  time intervals}}.
\newblock {\slshape \bibinfo{journal}{Journal of the ACM}}
  \bibinfo{volume}{38}(\bibinfo{number}{4}), p. \bibinfo{pages}{935–962},
  \doi{10.1145/115234.115351}.

\bibitemdeclare{phdthesis}{Kamp1968}
\bibitem{Kamp1968}
\bibinfo{author}{Johan Anthony~Wilem \surnamestart Kamp\surnameend}
  (\bibinfo{year}{1968}): \emph{\bibinfo{title}{Tense Logic and the Theory of
  Linear Order}}.
\newblock Ph.D. thesis, \bibinfo{school}{University of California, Los
  Angeles}.

\bibitemdeclare{article}{kowalski86}
\bibitem{kowalski86}
\bibinfo{author}{Robert \surnamestart Kowalski\surnameend} \&
  \bibinfo{author}{Marek \surnamestart Sergot\surnameend}
  (\bibinfo{year}{1986}): \emph{\bibinfo{title}{A logic-based calculus of
  events}}.
\newblock {\slshape \bibinfo{journal}{New Generation Computing}}
  \bibinfo{volume}{4}(\bibinfo{number}{1}), pp. \bibinfo{pages}{67--95},
  \doi{10.1007/BF03037383}.

\bibitemdeclare{inproceedings}{Krebs18}
\bibitem{Krebs18}
\bibinfo{author}{Andreas \surnamestart Krebs\surnameend}, \bibinfo{author}{Arne
  \surnamestart Meier\surnameend}, \bibinfo{author}{Jonni \surnamestart
  Virtema\surnameend} \& \bibinfo{author}{Martin \surnamestart
  Zimmermann\surnameend} (\bibinfo{year}{2018}): \emph{\bibinfo{title}{Team
  Semantics for the Specification and Verification of Hyperproperties}}.
\newblock In \bibinfo{editor}{Igor \surnamestart Potapov\surnameend},
  \bibinfo{editor}{Paul~G. \surnamestart Spirakis\surnameend} \&
  \bibinfo{editor}{James \surnamestart Worrell\surnameend}, editors: {\slshape
  \bibinfo{booktitle}{43rd International Symposium on Mathematical Foundations
  of Computer Science, {MFCS} 2018, August 27-31, 2018, Liverpool, {UK}}},
  {\slshape \bibinfo{series}{LIPIcs}} \bibinfo{volume}{117},
  \bibinfo{publisher}{Schloss Dagstuhl - Leibniz-Zentrum f{\"{u}}r Informatik},
  pp. \bibinfo{pages}{10:1--10:16}, \doi{10.4230/LIPIcs.MFCS.2018.10}.

\bibitemdeclare{inproceedings}{Lodaya_2000}
\bibitem{Lodaya_2000}
\bibinfo{author}{Kamal \surnamestart Lodaya\surnameend} (\bibinfo{year}{2000}):
  \emph{\bibinfo{title}{Sharpening the Undecidability of Interval Temporal
  Logic}}.
\newblock In \bibinfo{editor}{Jifeng \surnamestart He\surnameend} \&
  \bibinfo{editor}{Masahiko \surnamestart Sato\surnameend}, editors: {\slshape
  \bibinfo{booktitle}{Advances in Computing Science - {ASIAN} 2000, 6th Asian
  Computing Science Conference, Penang, Malaysia, November 25-27, 2000,
  Proceedings}}, {\slshape \bibinfo{series}{Lecture Notes in Computer Science}}
  \bibinfo{volume}{1961}, \bibinfo{publisher}{Springer}, pp.
  \bibinfo{pages}{290--298}, \doi{10.1007/3-540-44464-5\_21}.

\bibitemdeclare{inproceedings}{Mascle_2020}
\bibitem{Mascle_2020}
\bibinfo{author}{Corto \surnamestart Mascle\surnameend},
  \bibinfo{author}{Daniel \surnamestart Neider\surnameend},
  \bibinfo{author}{Maximilian \surnamestart Schwenger\surnameend},
  \bibinfo{author}{Paulo \surnamestart Tabuada\surnameend},
  \bibinfo{author}{Alexander \surnamestart Weinert\surnameend} \&
  \bibinfo{author}{Martin \surnamestart Zimmermann\surnameend}
  (\bibinfo{year}{2020}): \emph{\bibinfo{title}{From LTL to RLTL Monitoring:
  Improved Monitorability through Robust Semantics}}.
\newblock In: {\slshape \bibinfo{booktitle}{Proceedings of the 23rd
  International Conference on Hybrid Systems: Computation and Control}},
  \bibinfo{series}{HSCC '20}, \bibinfo{publisher}{Association for Computing
  Machinery}, \bibinfo{address}{New York, NY, USA}, pp.
  \bibinfo{pages}{170--204}, \doi{10.1145/3365365.3382197}.

\bibitemdeclare{inproceedings}{pitsikalis22a}
\bibitem{pitsikalis22a}
\bibinfo{author}{Manolis \surnamestart Pitsikalis\surnameend},
  \bibinfo{author}{Alexei \surnamestart Lisitsa\surnameend} \&
  \bibinfo{author}{Shan \surnamestart Luo\surnameend} (\bibinfo{year}{2021}):
  \emph{\bibinfo{title}{Representation and Processing of Instantaneous and
  Durative Temporal Phenomena}}.
\newblock In \bibinfo{editor}{Emanuele~De \surnamestart Angelis\surnameend} \&
  \bibinfo{editor}{Wim \surnamestart Vanhoof\surnameend}, editors: {\slshape
  \bibinfo{booktitle}{Logic-Based Program Synthesis and Transformation - 31st
  International Symposium, {LOPSTR} 2021, Tallinn, Estonia, September 7-8,
  2021, Proceedings}}, {\slshape \bibinfo{series}{Lecture Notes in Computer
  Science}} \bibinfo{volume}{13290}, \bibinfo{publisher}{Springer}, pp.
  \bibinfo{pages}{135--156}, \doi{10.1007/978-3-030-98869-2\_8}.

\bibitemdeclare{inproceedings}{pitsikalis22b}
\bibitem{pitsikalis22b}
\bibinfo{author}{Manolis \surnamestart Pitsikalis\surnameend},
  \bibinfo{author}{Alexei \surnamestart Lisitsa\surnameend},
  \bibinfo{author}{Patrick \surnamestart Totzke\surnameend} \&
  \bibinfo{author}{Simon \surnamestart Lee\surnameend} (\bibinfo{year}{2022}):
  \emph{\bibinfo{title}{Making Sense of Heterogeneous Maritime Data}}.
\newblock In: {\slshape \bibinfo{booktitle}{2022 23rd IEEE International
  Conference on Mobile Data Management (MDM)}}, pp. \bibinfo{pages}{401--406},
  \doi{10.1109/MDM55031.2022.00089}.

\bibitemdeclare{inproceedings}{Pnueli_1977}
\bibitem{Pnueli_1977}
\bibinfo{author}{Amir \surnamestart Pnueli\surnameend} (\bibinfo{year}{1977}):
  \emph{\bibinfo{title}{The temporal logic of programs}}.
\newblock In: {\slshape \bibinfo{booktitle}{18th Annual Symposium on
  Foundations of Computer Science (sfcs 1977)}}, \bibinfo{publisher}{IEEE},
  \bibinfo{address}{Providence, RI, USA}, p. \bibinfo{pages}{46–57},
  \doi{10.1109/SFCS.1977.32}.

\bibitemdeclare{article}{Rabinovich_2014}
\bibitem{Rabinovich_2014}
\bibinfo{author}{Alexander \surnamestart Rabinovich\surnameend}
  (\bibinfo{year}{2014}): \emph{\bibinfo{title}{{A Proof of Kamp's theorem}}}.
\newblock {\slshape \bibinfo{journal}{{Logical Methods in Computer Science}}}
  \bibinfo{volume}{{Volume 10, Issue 1}}, \doi{10.2168/LMCS-10(1:14)2014}.

\bibitemdeclare{article}{RAY2019}
\bibitem{RAY2019}
\bibinfo{author}{Cyril \surnamestart Ray\surnameend}, \bibinfo{author}{Richard
  \surnamestart Dréo\surnameend}, \bibinfo{author}{Elena \surnamestart
  Camossi\surnameend}, \bibinfo{author}{Anne-Laure \surnamestart
  Jousselme\surnameend} \& \bibinfo{author}{Clément \surnamestart
  Iphar\surnameend} (\bibinfo{year}{2019}): \emph{\bibinfo{title}{Heterogeneous
  integrated dataset for Maritime Intelligence, surveillance, and
  reconnaissance}}.
\newblock {\slshape \bibinfo{journal}{Data in Brief}} \bibinfo{volume}{25}, p.
  \bibinfo{pages}{104141}, \doi{10.1016/j.dib.2019.104141}.

\bibitemdeclare{article}{Venema_1990}
\bibitem{Venema_1990}
\bibinfo{author}{Yde \surnamestart Venema\surnameend} (\bibinfo{year}{1990}):
  \emph{\bibinfo{title}{Expressiveness and completeness of an interval tense
  logic.}}
\newblock {\slshape \bibinfo{journal}{Notre Dame Journal of Formal Logic}}
  \bibinfo{volume}{31}(\bibinfo{number}{4}), \doi{10.1305/ndjfl/1093635589}.

\bibitemdeclare{incollection}{OHRSTROM2006447}
\bibitem{OHRSTROM2006447}
\bibinfo{author}{Peter \surnamestart Øhrstrøm\surnameend} \&
  \bibinfo{author}{Per \surnamestart Hasle\surnameend} (\bibinfo{year}{2006}):
  \emph{\bibinfo{title}{Modern temporal logic: The philosophical background}}.
\newblock In \bibinfo{editor}{Dov~M. \surnamestart Gabbay\surnameend} \&
  \bibinfo{editor}{John \surnamestart Woods\surnameend}, editors: {\slshape
  \bibinfo{booktitle}{Logic and the Modalities in the Twentieth Century}},
  {\slshape \bibinfo{series}{Handbook of the History of
  Logic}}~\bibinfo{volume}{7}, \bibinfo{publisher}{North-Holland}, pp.
  \bibinfo{pages}{447--498}, \doi{10.1016/S1874-5857(06)80032-4}.

\end{thebibliography}
\end{document}